\documentclass[11pt]{article}
\usepackage{amsthm}
\usepackage{amsmath}
\usepackage{amssymb}
\usepackage{graphicx}
\usepackage{tikz}
\usepackage{cases}
\usepackage{url}
\usetikzlibrary{arrows,positioning,fit,shapes,calc}

\newtheorem{Thm}{Theorem}
\newtheorem{Lem}[Thm]{Lemma}
\newtheorem{Def}[Thm]{Definition}
\newtheorem{Pro}[Thm]{Proposition}
\newtheorem{Cor}[Thm]{Corollary}
\newtheorem{Rem}[Thm]{Remark}

\newcommand{\bfm}[1]{\mbox{\boldmath ${#1}$}}
\newcommand{\nonum}{\nonumber \\}
\newcommand{\beqa}{\begin{eqnarray}}
\newcommand{\eeqa}[1]{\label{#1}\end{eqnarray}}
\newcommand{\beq}{\begin{equation}}
\newcommand{\eeq}[1]{\label{#1}\end{equation}}
\newcommand{\Grad}{\nabla}
\newcommand{\Div}{\nabla \cdot}

\newcommand{\Imag}{\operatorname{ Im}}
\newcommand{\Real}{\operatorname{ Re}}

\newcommand{\Ga}{\alpha}

\newcommand{\Ge}{\epsilon}
\newcommand{\Gve}{\varepsilon}

\newcommand{\Go}{\omega}

\newcommand{\GD}{\Delta}

\newcommand{\GO}{\Omega}

\newcommand{\BGa}{\bfm\alpha}

\newcommand{\BGe}{\bfm\epsilon}
\newcommand{\BGve}{\bfm\varepsilon}

\newcommand{\md}{{\mathrm{ d}}}
\newcommand{\mm}{{\mathrm{ m}}}
\newcommand{\me}{{\mathrm{ e}}}
\newcommand{\bbR}{{\mathbb{ R}}}
\newcommand{\bbC}{{\mathbb{C}}}

\newcommand{\CD}{{\cal D}}
\newcommand{\CE}{{\cal E}}

\newcommand{\CM}{{\cal M}}

\newcommand{\CO}{{\cal O}}
\newcommand{\CP}{{\cal P}}

\def\ii{{\rm i}}

\def\Bf{{\bf f}}

\def\Bn{{\bf n}}

\def\Bp{{\bf p}}

\def\Bx{{\bf x}}
\def\By{{\bf y}}

\def\BA{{\bf A}}
\def\BB{{\bf B}}

\def\BD{{\bf D}}
\def\BE{{\bf E}}

\def\BH{{\bf H}}
\def\BI{{\bf I}}

\def\BL{{\bf L}}
\def\BM{{\bf M}}

\def\BP{{\bf P}}

\def\BU{{\bf U}}
\def\BV{{\bf V}}

\def \AA {{\mathbb A}}

\def \ba {\begin{array}}
\def \ea {\end{array}}

\begin{document}
\vspace{-1in}
\title{Bounds on Herglotz functions and fundamental limits of broadband passive quasi-static cloaking}
\author{Maxence Cassier\\
\small{Department of Mathematics, University of Utah, Salt Lake City UT 84112, USA}\\
\small{(cassier@math.utah.edu)} \\
\\
Graeme W. Milton\\
\small{Department of Mathematics, University of Utah, Salt Lake City UT 84112, USA} \\
\small{(milton@math.utah.edu)}}

\date{}
\maketitle
\begin{abstract}
Using a sum rule, we derive new bounds on Herglotz functions that generalize those given in \cite{Gustafsson:2010:SRP,Bernland:2011:SRC}. These bounds apply to a wide class of linear passive systems such as electromagnetic passive materials. Among these bounds, we describe the optimal ones and also discuss their meaning in various physical situations like in the case of a transparency window, where we exhibit sharp bounds. Then, we apply these bounds in the context of broadband passive cloaking in the quasi-static regime to negatively answer the following challenging question: is it possible to construct a passive cloaking device that cloaks an object over a whole frequency band? Our rigorous approach, although limited to quasi-statics, gives quantitative limitations on the cloaking effect over a finite frequency range by providing inequalities on the polarizabilty tensor associated with the cloaking device. We emphasize that our results hold for a cloak or object of any geometrical shape.
\end{abstract}
\vskip2mm

\section{Introduction}
%%%%%%%%%%%%%%%%%%%%%%%%%%%%%%%%%%%%%%%%%%%%%%%%%%%%%%%%%%%%%%%%%%%%%%%%%%
\setcounter{equation}{0}

For many years it has been know that there exist inclusions that can be invisible to certain applied fields. These are generally known as neutral inclusions (see, for example, the references in Section 7.11 in  \cite{Milton:2002:TC},
and see also the more recent citations of these papers)
and references therein) and a specific example are the invisible bodies
of Kerker \cite{Kerker:1975:IB} , that are in fact coated confocal ellipsoids, which are invisible to long wavelength fields. More intriguing are the cylindrical shells of NIcorovici, McPhedran, and Milton  \cite{Nicorovici:1994:ODP} having (at a fixed frequency) a relative
permittivity of $-1$, surrounded by material having 
a reative permittivity of $1$, that are invisible to any polynomial quastistatic applied field, and the coated spheres of Al\'u and Engheta \cite{Alu:2005:ATP},
that are invisible at a specific frequency. For conductivity and fixed frequency electromagnetism
Tartar (in a private communication to Kohn and Vogelius \cite{Kohn:1984:IUC})
and Dolin \cite{Dolin:1961:PCT} recognised that one could create a wide class of invisible inclusions with anisotropic moduli by transformation conductivity and transformation optics. A subsequent key idea of Greenleaf, Lassas, and Uhlmann \cite{Greenleaf:2003:ACC,Greenleaf:2003:NCI}
was that one could create a cloak for conductivity (and hence single frequency quasi-statics) by using singular transformations that created a "quiet zone" where no field penetrated, and
hence where one could place an object without disturbing the surrounding current field. The next development was the recognition by Milton and Nicorovici \cite{Milton:2006:CEA}
that cloaking due to anomalous resonance could cloak 
(at least in two-dimensional quasi-statics, though some results were also obtained in three-dimensions and at finite frequency) an arbitary finite number of polarizable dipoles: this had the fascinating feature that the
cloaking region lay {\it outside} the cloaking device. It was perhaps the first paper where the word cloaking appeared in the scientific literature, outside computer science. 
Shortly afterwards, papers appeared by Leonhardt \cite{Leonhardt:2006:OCM} and Pendry, Schurig, and Smith \cite{Pendry:2006:CEM}
using transformation ideas to obtain cloaking for geometric optics and Maxwell's equations at fixed frequencies. These three papers,
of Milton and Nicorovici, Leonhardt, and Pendry, Schurig and Smith generated considerable media attention, and also stimulated a lot of subsequent scientific development, 
both on cloaking due to anomalous resonance \cite{Nicorovici:2007:QCT,Bruno:2007:SCS,Milton:2008:SFG,Nicorovici:2008:FWC,McPhedran:2009:CPR,Bouchitte:2010:CSO,Nicorovici:2011:RLD,Kohn:2012:VPC, Ammari:2013:ALR,Ammari:2013:STN,Ammari:2013:STNII, Kettunen:2014:AEA,Ando:2015:SPN,Nguyen:2015:CAL, Meklachi:2016:SAL, Li:2016:QCA, Nguyen:2016:CVA, Onofrei:2016:ALR,Nguyen:2016:EWP}
and on transformation based cloaking (see, for example, the reviews \cite{Alu:2008:PMC} and \cite{Greenleaf:2009:CDE}).
Other sorts of cloaking were developed too, including cloaking due to complementary media \cite{Lai:2009:CMI}, that has anomalous resonance 
as its mechanism \cite{Nguyen:2016:CCM, Nguyen:2015:CAL1}, and active cloaks 
\cite{Miller:2006:PC, Vasquez:2009:AEC, Vasquez:2009:BEC, Vasquez:2012:MAT, Onofrei:2012:AMF, Norris:2012:SAA, Slvanayagam:2012:AEC,Norris:2014:AEC} where sources tailored to the incoming signal, and sometimes also tailored to the body to be cloaked \cite{ONeill:2015:ACI, ONeill:2016:ACR}, create a cloak, yet do not significantly radiate. 
There is no theoretical difficulty in creating broadband active cloaks: each frequency can be cloaked separately and sources can then be designed that
superimpose the contributions from the different frequencies. A good example is the cloaking of an object from an incoming pulse in the animation movies in \cite{Vasquez:2009:BEC}.

Here our focus is on finding limitations to broadband cloaking for passive quasi-static cloaks.
Unlike active cloaks which require energy sources to activate them, passive cloaks perform cloaking only by the physical properties of the passive material which constitutes the cloak.
In the context of transformation based cloaking it has long been recognized that a cloak that guides waves around an object has the inherent limitation that a pulse signal (hence containing many frequencies) travelling on a ray cannot travel faster than the speed of light and therefore, if the ray goes around the body, the pulse will arrive later compared to a pulse that 
travels in a straight line at the speed of light. However we would like some more explicit quantitative bounds that limit cloaking, in particular over a specific frequency interval. Anomalous resonance uses materials with
a negative dielectric constant and transformation based cloaks use materials with relative electrical permittivities (relative compared to the surrounding medium) less than one. Thus if the surrounding medium has the electrical permittivity of free space there should necessarily be some variation of the moduli with frequency, i.e., dispersion. While some experiments report broadband cloaking it is to emphasized that the surrounding medium is silicon, and this makes
it possible to achieve a relative electrical permittivity that is less than 1 that is almost frequency independent. It seems that a clue to establishing broadband limitations to cloaking is to use bounds limiting the minimal
dispersion in the component materials. For geometric optics Leonhardt and Tyc \cite{Leonhardt:2009:BIN} show one can get broadband cloaking by ingeneous transformations from non-Euclidean geometries to
Euclidean ones. (Such transformations are okay for geometric optics, but generally do not preserve the form of the time-harmonic Maxwell's equations.)

The main tool used to derive our bounds is to follow the idea developed in the analytic method introduced in \cite{Bergman:1978:APC,Milton:1979:TST}, justified in \cite{Milton:1981:BCP} and proved in \cite{Golden:1983:BEP}. In other words, to use the analytic properties of physical quantities (like the dielectric permittivity and  the magnetic permeability in electromagnetism) which define the constitutive laws of the medium in the frequency domain. These properties are the counterpart of causality and passivity of time-dependent passive linear materials. Mathematically speaking, it is directly linked to the existence of a Herglotz and/or a Stieltjes function which characterizes the behavior of the system in the frequency domain \cite{Nussenzveig:1972:CDR, Cessenat:1996:MME, Milton:2002:TC,Zemanian:2005:RTC,Bernland:2011:SRC}. This analytic method under various forms has been widely applied to study physical properties of passive electromagnetic media in different contexts: to bound the dielectric permittivity with respect to the frequency \cite{Bernland:2011:SRC,Gustafsson:2010:SRP}, to evaluate the resolution of a perfect lens on a finite bandwidth \cite{Lind-Johansen:2009:PLF}, to derive scattering limits as for instance upper bounds on the total extinction cross-section \cite{Purcell:1969:AEL,Gustafsson:2010:TDS,Miller:2014:FLE} or to provide quantitative limits to speed light propagation in dispersive media \cite{Welters:2014:SLL} .
In this paper, one wants to use such a method to derive bounds on the polarizability tensor associated with a cloaking device. This tensor is defined as a $3\times 3$ complex-valued matrix function of the frequency \cite{Jackson:1999:CE,Milton:2002:TC} which characterizes the main contribution of the far field of the scattered wave due to a cloaking device in the quasi-static regime of Maxwell's equations. Therefore if it vanishes at a frequency $\Go$, one says that the obstacle is cloaked at $\Go$ for a far observer. We prove in this paper that is not possible on a whole frequency band and derive inequalities to quantify this phenomenon.

Related to the question of broadband passive cloaking, we mention that Monticone and Al\'u \cite{Mon:2014:PBE} show that one cannot perform passive cloaking
on the whole frequency spectrum by deriving a global bound on the scattering cross-section. More recently  \cite{Mon:2016:IEP} they use electrical circuit analogies,
to bound the scattering cross section over a finite frequency range for planar objects. Another interesting point was developed by Hashemi, Qiu, McCauley, Joannopoulos and Johnson \cite{Hashemi:2012:DBP} who demonstrate for the particular case of Lorentz dispersion models  that broadband  passive cloaking is limited by the obstacle characteristic size.  Here, the bounds that we derive have the great advantage of neither assuming the geometrical shape of the object, or cloak, nor the dispersion models of the cloak. In fact the object could even lie outside the cloak.
Moreover, they involve the size of the frequency bandwidth. While they are limited to quasi-statics, they apply to cloaking due to anomalous resonance, transformation based cloaking, 
and cloaking due to complementary media and in fact to any quasi-static passive cloaking device.

The paper is organized as follows. In section \ref{sec-boundanalycomp}, we first derive, using complex analysis, general bounds that are applicable to a broadband class of passive linear systems including electromagnetic passive media. More precisely, for an electromagnetic passive material, the standard notions of passivity and causality are introduced and this leads to four constraints on the dielectric permittivity and magnetic permeability behaviors seen as complex valued functions of the frequency. To develop our bounds in the general framework of linear passive system, we reformulate these four constraints as assumptions on a abstract complex-valued function $f$. Then, we briefly recall some basic notions on Stieltjes and Herglotz functions which are used throughout the paper. Our next step is to construct a Stieltjes and a Herglotz function associated with $f$.  Afterwards, using the sum rules derived in \cite{Bernland:2011:SRC} for Herglotz functions, we derive bounds parametrized by a set a of probability measures that generalized the bounds of \cite{Gustafsson:2010:SRP,Bernland:2011:SRC}. Then, we  prove that among these bounds, the ones that are optimal are obtained using Dirac measures (see Theorem  \ref{thm.mesopt}). Using such measures in the case of a transparency window (which physically means that the material is lossless on the considered frequency range), we recover a bound similar to the ones derived in \cite{Milton:1997:FFR,Yaghjian:2006:PWS} which is sharp for Drude type models. We show that this last bound can be also easily established by another approach based on Kramers--Kronig relations. We finally explore the case of lossy material and recovers by our approach a bound similar to the ones derived in \cite{Gustafsson:2010:SRP,Bernland:2011:SRC}. The section \ref{sec.boundcloak} of the paper is devoted to the applications of the previous bounds to the broadband passive cloaking question for the quasi-static approximation of Maxwell's equations. We first mathematically reformulate our cloaking problem in a rigorous functional framework and shows that the bounds derived in section \ref{sec-boundanalycomp} apply to the polarizability tensor associated with a passive cloaking device. Finally, we show that it is not possible to construct a passive cloak that achieves broadband cloaking on a finite range of frequencies and discuss the meaning of our bounds as fundamental limits of the cloaking effect in various situations like a transparency window or the general case of a lossy material. 

%%%%%%%%%%%%%%%%%%%%%%%%%%%%%%%%%%%%%%%%%%%%%%%%%%%%%%%%%%%%%%%%%%%%%%%%%
\section{Bounds on Herglotz functions}\label{sec-boundanalycomp}
\subsection{Characterization of passive electromagnetic media}\label{First part}
In this subsection, one introduces the standard notions of causality and passivity for linear time-dependent Maxwell's equations and their counterparts in the frequency domain. For simplicity, we are dealing here with an isotropic homogeneous material which fills a bounded domain $\Omega\subset \bbR^3$, but one can derive such properties in the general setting of anisotropic and inhomogeneous materials. For more details, we refer to \cite{Nussenzveig:1972:CDR,Landau:1984:ECM,Cessenat:1996:MME,Milton:2002:TC,Zemanian:2005:RTC,Welters:2014:SLL}.

We denote respectively by $\BD$ and $\BB$ the electric and magnetic inductions, by  $\BE $ and $\BH$ the electric and magnetic fields, the evolution of $(\BE,  \BD, \BH, \BB)$ in $\Omega$ is governed (on the absence of a current density source) by the macroscopic Maxwell's equations:
\begin{equation}\label{eq.Maxwell}
   \partial_t \BD -\nabla \times  \BH =0 \quad\mbox{and}\quad \partial_t \BB +\nabla \times \BE = 0,
\end{equation} 
which must be supplemented by the constitutive laws of the material involving two additional unknowns the electric and magnetic polarizations $\BP$ and $\BM$:
\begin{equation}\label{eq.constiutivelaws}
\BD= \BGve_0 \BE +\BP \ \mbox{ with }  \BP= \Gve_0\,  \chi_E \star_t \BE  \ \mbox{ and } \BB= \mu_0 \BH +\BM \mbox{ with }  \BM= \mu_0 \, \chi_M \star_t \BH.
\end{equation}
The constants $\Gve_0$ and $\mu_0$ stand here for the permittivity and permeability of the vacuum. 
The constitutive laws express the  relations between $(\BD,\BE)$ and $(\BB,\BH)$ via a convolution in time with the electrical and magnetic susceptibility $\chi_E$ and $\chi_M$, defined here as scalar time-dependent functions which characterize the electromagnetic behavior of the material.

We assume  here for simplicity that $\chi_E $ and $\chi_M\in L^1(\bbR_t)$, the space of integrable functions with respect to the time variable. In a more general setting, one can consider them as tempered distributions (see \cite{Cessenat:1996:MME,Zemanian:2005:RTC}). We suppose also that $\BE$, $\BH$, $\partial_t \BE$ and  $\partial_t \BH$ are in $\BL^2(\bbR_t,\BL^2(\Omega) )$.
Hence, as $(\BE,  \BD, \BH, \BB)$ satisfy (\ref{eq.Maxwell}) and (\ref{eq.constiutivelaws}), one deduces with such hypothesis that $\BD,\, \BB, \, \partial_t \BD, \,\partial_t \BB$,  $\nabla \times \BE$ and $\nabla \times \BH$ are also in $\BL^2(\bbR_t,\BL^2(\Omega) )$. In this functional framework, one introduces four standard properties which model the constitutive laws of electromagnetic passive linear systems in the frequency domain.

$\bullet$ A material is said to be causal if the fields $\BE(\cdot,t)$ and $\BH(\cdot,t)$ cannot influence the inductions $\BD(\cdot,t')$ and $\BB(\cdot,t')$ for $t'<t$. This condition implies that the functions $\chi_E $ and $\chi_M$ are supported in $\bbR^{+}$. To see the counterpart of the causality in the frequency domain, one defines the Laplace-Fourier transform of a function of $f\in L^1(\bbR_t)$ supported in $\bbR^{+}$ by
\begin{equation} \label{eq.Laplacetransform}
\hat{f}(\omega)=\int_{\bbR^+}f(t) \, \me^{\ii   \Go t}\md t,  \, \forall \Go \in \operatorname{cl}{\bbC^{+}},
\end{equation}
where $\operatorname{cl}{\bbC^{+}}:=\{\Go \in \bbC \mid \Imag(\Go)\geq 0\}$ stands for topological closure of the complex upper half-plane $\bbC^{+}:=\{\Go \in \bbC \mid \Imag(\Go)>0 \}$. We point out that the Laplace-Fourier transform coincides with the Fourier transform for real frequency $\Go$, that is why we use in the following the same notation for both transforms. Classically, applying the Fourier transform to  (\ref{eq.constiutivelaws}) for real $\Go$ leads to the  well-known expression for the constitutive laws (\ref{eq.constiutivelaws}) in the frequency domain:
$$
\hat{\BD}(\Go)=\Gve(\Go) \hat{\BE}(\Go) \mbox{ with } \Gve(\Go)=\Gve_0 \big(1+\hat{\chi}_{E}(\Go) \big) \mbox{ and }  \hat{\BB}(\Go)=\mu(\Go)  \hat{\BH}(\Go) \mbox{ with } \mu(\Go)=\mu_0 \big(1+\hat{\chi}_M(\Go)\big),
$$
where $\Gve(\Go)$ and $\mu(\Go)$ stand for the dielectric permittivity and the magnetic permeability of the material. Now, as $\chi_E$ and $\chi_M \in L^1(\bbR_t)$ are compactly supported in $\bbR^{+}$, one deduces easily that their Laplace Fourier transforms $\hat{\chi}_E$ and $\hat{\chi}_M$ are analytic in the upper half-plane $\bbC^{+}$ and continuous on $\operatorname{cl} \bbC^{+}$. Thus, $\Gve=\Gve_0(1+\hat{\chi}_{E})$ and $\mu=\mu_0(1+\hat{\chi}_{M})$ share the same regularity. 

$\bullet$ Furthermore, by applying the Riemann-Lebesgue theorem  (since $\chi_E $ and $\chi_M\in L^1(\bbR_t)$), one has that $\hat{\chi}_E$ and  $\hat{\chi}_M$ tend to $0$, as $|\Go|\to \infty$ in $\operatorname{cl} \bbC^{+}$. Hence,  we have $$\Gve(\Go)\to \Gve_0 \mbox{ and } \mu(\Go)\to  \mu_0, \mbox{ as  } |\Go|\to \infty  \mbox{ in }\operatorname{cl} \bbC^{+} .$$
In other words, the material behaves as the vacuum for high frequencies.

$\bullet$ As $\chi_E$ and $\chi_M$ are real functions, it implies that their Laplace-Fourier transforms defined by (\ref{eq.Laplacetransform}) satisfisfy the following ``symmetry'' relations:
\begin{equation}\label{eq.symfreq}
\hat{\chi}_E(-\overline{\Go})=\overline{\hat{\chi}_E(\Go)} \ \mbox{ and } \ \hat{\chi}_M(-\overline{\Go})=\overline{\hat{\chi}_M(\Go)},\quad\forall \Go \in \operatorname{cl} \bbC^{+},
\end{equation}
and thus the same relation holds for the functions $\Gve$ and $\mu$.

$\bullet$ The passivity assumption is expressed as the following (see \cite{Bernland:2011:SRC,Cessenat:1996:MME,Landau:1984:ECM,Milton:2002:TC,Welters:2014:SLL}):
\begin{equation}\label{eq.passivity}
\quad \CE_a(t)=\int_{-\infty}^{t}  \int_{\Omega}  \partial_t \BD(\Bx,s) \cdot \BE(\Bx,s)+  \partial_t \BB(\Bx,s) \cdot \BH(\Bx,s) \, \md \Bx \, \md s \geq 0, \forall t \in \bbR
\end{equation}
and holds for any fields $(\BE,\BH)$ such that
\begin{equation}\label{eq.regularity}
\BE,\, \BH \in \BL^2(\bbR_t,\BL^2(\Omega) ) \mbox{ and } \partial_t \BE,\, \partial_t \BH \in \BL^2(\bbR_t,\BL^2(\Omega) ).
\end{equation}
This assumption imposes physically that at each time, the amount of  electromagnetic energy $ \CE_a(t)$ transferred to the material by Joule effect or absorption, that is by electric and/or magnetic loss is positive. 
By virtue of the Plancherel theorem and the constitutive laws (\ref{eq.constiutivelaws}), the passivity assumption (\ref{eq.passivity}) applied to $t=\infty$ yields the following inequality in the frequency domain: 
$$
\CE_a(\infty)=\frac{1}{2\pi} \Real \left(\int_{\bbR} \int_{\Omega} - \ii  \omega \Big(  \Gve_0 \big(1+\hat{\chi}_E(\Go)\big) |\hat{\BE}(\Bx,\Go)|^2 + \mu_0\big(1+ \,\hat{\chi}_M(\Go)\big) |\hat{\BH}(\Bx,\Go)|^2 \Big) \md \Bx \,\md \Go \right) \geq 0 
 $$ 
which can be rewritten as 
\begin{equation}\label{eq.passineqfreq}
\CE_a(\infty)=\frac{1}{2\pi}\int_{\bbR} \int_{\Omega}   \Go \Imag \Gve(\Go) |\hat{\BE}(\Bx,\Go)|^2 +  \Go \Imag \mu(\Go) |\hat{\BH}(\Bx,\Go)|^2  \md \Bx \,\md \Go  \geq 0.
\end{equation}
Hence, as the last inequality holds for any fields $\BE$ and $\BH$ which satisfy the conditions (\ref{eq.regularity}) in the time-domain, it is straightforward (using a proof by contradiction) to show that it implies that $\Go \Imag \Gve(\Go)\geq 0$ and $\Go \Imag \mu(\Go)\geq 0$,  for all real  frequency $\Go$. These latter conditions turn  out to be equivalent, by (\ref{eq.symfreq}), to
\begin{equation}\label{eq.passivityharm}
\Imag \Gve(\Go)\geq 0 \ \mbox{ and } \ \Imag \mu(\Go)\geq 0, \, \forall \Go \in \bbR^{+},
\end{equation}
that is referred to the characterization of passivity in the frequency domain \cite{ Landau:1984:ECM,Cessenat:1996:MME,Milton:2002:TC}. 

Reciprocally, the condition (\ref{eq.passivityharm}) and the fact that $\Gve$ and $\mu$ are bounded, continuous functions  (since $\chi_E $ and $\chi_M\in L^1(\bbR_t)$) satisfying (\ref{eq.symfreq}) on $\bbR_{\Go}$ imply, in particular, that inequality (\ref{eq.passineqfreq}) holds  for any $\hat{\BE },\hat{\BH}\in \BL^2(\bbR_{\Go},\BL^2(\Omega) )$ such that $\BE$ and $\BH\in \CD\big((-\infty, t), \BL^2(\Omega)\big)$ (where $\CD\big( (-\infty, t), \BL^2(\Omega)\big)$ refers to the space of bump functions of $(-\infty, t)$ valued in $\BL^2(\Omega)$). Hence by Plancherel's theorem, one obtains that passivity assumption (\ref{eq.passivity}) holds at any fixed time $t\in \bbR$ and for any $\BE$ and $\BH\in$ $\CD\big((-\infty, t) , \BL^2(\Omega)\big)$. Finally, one extends by a density  argument this relation to any fields $\BE$ and $\BH$ satisfying (\ref{eq.regularity}). Thus, (\ref{eq.passivityharm}) is equivalent to (\ref{eq.passivity}).

The aim of this section is to derive in a general framework a bound for a function $f: \operatorname{cl}{\bbC^{+}} \mapsto \bbC$ which satisfies the
following hypotheses:
\begin{itemize}
\item H1: $f$ is analytic on the upper half plane $\bbC^{+}$ and continuous on $\operatorname{cl} \bbC^{+}$,
\item H2: $f
(z) \to  f_{\infty}>0$, when $\left|z\right|\to \infty$ in $\operatorname{cl} \bbC^{+}$,

\item H3: $f$ satisfies $f(-\overline{z})=\overline{f(z)}, \quad\forall z \in \operatorname{cl} \bbC^{+}$,
\item H4: $\Imag f(z)\geq 0$ for all $z\in \bbR^{+}$ (passivity).

\end{itemize}
described above for $f=\Gve$ or $f=\mu$ as function of the frequency $\Go$.
More generally, these hypotheses characterize the frequency behavior of passive linear systems \cite{Zemanian:2005:RTC,Bernland:2011:SRC}.
They are satisfied by the permittivity and the permeability but also by other physical quantities such as the polarizability tensor in the quasi-static regime (as it will be proved in subsection \ref{sub.analypola}), the acoustic \cite{Norris:2015:AIE} and electromagnetic  \cite{Nussenzveig:1972:CDR,Gustafsson:2010:TDS} forward scattering amplitudes and the shear and bulk modulus in elasticity \cite{Bonifasi:2009:Aar}. Thus, the bounds we develop in this first part, in this general setting, apply to all these physical parameters and constrain their behavior in the frequency domain.

\subsection{Review of some Herglotz and Stieltjes functions properties}
Mathematically, the hypotheses H1-4 on the function $f$ are linked to the existence of a Stieltjes and a Herglotz function associated with $f$. Stieltjes and Herglotz functions have been extensively used in the study of electromagnetic materials' behavior (see for instance \cite{Milton:1997:FFR, Milton:2002:TC,Gustafsson:2010:SRP,Bernland:2011:SRC,Welters:2014:SLL}).
The aim of this subsection is to recall briefly some properties about these functions that we use in the following to derive our bounds. For more details, we refer to \cite{Nevanlinna:1922:AEF,Baker:1981:PAB,Gesztesy:2000:MVH,Berg:2008:SPBS, Bernland:2011:SRC}.

\begin{Def}
An analytic function $h:\bbC^{+}\to \bbC$ is a Herglotz function (also called Pick or Nevanlinna function) if
$$
\Imag h(z)\geq 0, \ \forall z \in \bbC^{+}. 
$$
\end{Def}

A particular and useful property of Herglotz functions is the following representation theorem due to Nevanlinna \cite{Nevanlinna:1922:AEF}.

\begin{Thm}\label{thm.Herg}
A necessary and sufficient condition for $h$ to be a Herglotz function is given by the following representation:
\begin{equation}\label{eq.defhergl}
h(z)=\alpha \, z+\beta + \displaystyle \int_{\bbR} \left(  \frac{1}{\xi-z}- \frac{\xi}{1+\xi^2}\right)\md \mm( \xi), \ \mbox{ for } \Imag(z)>0,
\end{equation}
where $\alpha\in\bbR^{+}$, $\beta\in \bbR$ and $\mm$ is a positive regular Borel measure for which   $\int_{\bbR} \
 \md \mm(\xi)/(1+\xi^2)$ is finite.  
In particular if the integral $\int_{\bbR} \xi \, \md \mm(\xi)/(1+\xi^2)$ is also finite, then we can rewrite  the relation (\ref{eq.defhergl}) as:
\begin{equation*}\label{eq.herg2}
h(z)=\alpha \, z+\gamma  + \displaystyle \int_{\bbR} \frac{\md  \mm( \xi) }{\xi-z}  \ \quad \mbox{ with } \gamma=\beta- \int_{\bbR}  \frac{\xi \,\md \mm( \xi)}{1+\xi^2}\in \bbR.
\end{equation*}
\end{Thm}

Moreover, for a given Herglotz function $h$, the triple $(\alpha, \beta, \mm)$ is uniquely defined by the following corollary.

\begin{Cor}\label{cor.Herg}
Let $h$ be a Herglotz function defined by its representation (\ref{eq.defhergl}), then we have:
\begin{eqnarray}
&& \alpha= \lim_{y\to +\infty}\displaystyle \frac{h(\ii y)}{\ii y}, \ \beta=\Real h(\ii), \nonumber \\[10pt] 
\mbox{ and } \forall  [a,b] \subset \bbR, && \frac{\mm([a,b])+\mm((a,b))}{2} =\lim_{y\to 0^{+}} \frac{1}{\pi} \int_{a}^{b}  \Imag h(x+\ii y) \md x \label{eq.mesHerg}.
\end{eqnarray}
\end{Cor}

We now introduce for any $\theta\in (0,\pi/2)$ the Stolz domain $D_{\theta}$ defined by: $$D_{\theta}=\{z\in \bbC \mid  \theta \leq \operatorname{arg}(z) \leq \pi-\theta  \}.$$
The representation theorem \ref{eq.defhergl} implies (see \cite{Bernland:2011:SRC}) that a Herglotz function satisfies the following asymptotics in $D_{\theta}$  for all $\theta\in (0,\pi/2)$:
\begin{equation}\label{eq.asymherglotzfunction}
h(z)=-\mm(\{0\}) z^{-1}+o(z^{-1} ) \ \mbox{ as } |z| \to  0 \ \mbox{ and } \ h(z)=\alpha \, z+o(z) \mbox{ as } |z| \to + \infty.
\end{equation}
In other words, an Herglotz function grows at most as rapidly $z$ when $ |z|$  tends to $+\infty$ and cannot be more singular than $z^{-1}$ when $ |z|$ tends to $0$. 

We will conclude this review of Herglotz functions by a last identity: the so-called sum rule (see \cite{Bernland:2011:SRC}) which is a fundamental tool to derive quantitative bounds on passive systems.

\begin{Pro}\label{pro.sumrules}
Let $h$ be a Herglotz function which admits the following asymptotic expansions in $D_{\theta}$ for all $\theta\in (0,\pi/2)$:
\begin{eqnarray*}
&&h(z)=a_{-1} \,z^{-1}+o(z^{-1}) \ \mbox{ as } |z| \to  0, \\[10pt] 
\mbox{ and } && h(z)=  b_{-1}\, z^{-1}+o(z^{-1})\ \mbox{ as }  |z|  \to +\infty.
\end{eqnarray*}
with $a_{-1}$ and $b_{-1} \in \bbR$. Then the following identity holds 
\begin{equation}\label{eq.sumrules}
\lim_{\eta \to 0^{+}}\lim_{y\to 0^{+}}\frac{1}{\pi} \int_{\eta<|x|<\eta^{-1}} \Imag h (x+\ii y) \, \md x=a_{-1}-b_{-1}.
\end{equation}
\end{Pro}
We now introduce Stieltjes functions: another famous class of analytic functions, closely related to Herglotz functions.
\begin{Def}
A Stieltjes function is an analytic function $g:\bbC\setminus{\bbR}^{-} \to \bbC$ which satisfies:
$$
\Imag g(z)\leq 0 \ \,  \forall z \in \bbC^{+} \ \mbox{ and } \ g(x)\geq 0 \  \mbox{ for } x > 0.
$$
\end{Def}

Like Herglotz functions, Stieltjes functions are characterized by a representation theorem.

\begin{Thm}\label{thm.Stielt}
A necessary and sufficient condition for $g$ to be a Stieltjes function is given by the following representation:
\begin{equation*}
g(z)=\alpha+ \displaystyle \int_{\bbR^{+}}  \frac{\md \mm( \xi)}{\xi+z} \quad \forall z\in \bbC\setminus \bbR^{-},
\end{equation*}
where $\alpha= \lim \limits_{|z|\to +\infty} g(z) \in \bbR^{+}$ and $\mm$ is a positive regular Borel measure, uniquely defined, for which  $\int_{\bbR^{+}} \
 \md \mm(\xi)/(1+\xi)$ is finite.  
\end{Thm}

\begin{Rem}
An easy connection can be made between Herglotz and Stieljes function. Thanks to the representation Theorems \ref{thm.Herg} and \ref{thm.Stielt}, we note that if $g$ is a Stieltjes function, the function $h$ defined by $h(z)=g(-z)$ is an Herglotz function whose measure $\mm$ has a support included in $\bbR^{+}$ in the relation (\ref{eq.defhergl}). Another connection between Herglotz and Stieltjes functions is given in the next subsection by Corollary \ref{cor.herg}.
\end{Rem}

\subsection{Construction of a Stieltjes function associated with $f$}\label{sub-secconstrHerg}
In this paragraph, we construct a Stieltjes function associated with the function $f$.
For that, we first establish with the following lemma some information about the sign of the imaginary part of the function $f$.

\begin{Lem}\label{lem.partimag}
If a function $f$ satisfies the hypotheses H1-4, then 
\begin{equation}\label{eq.imag}
 \pm \Imag f(z)\geq 0, \ \forall z\in \bbC^{+} \  \mbox{ such that } \pm \Real z\geq 0.  
\end{equation}
Moreover, if $f$ is not a constant function,  the inequalities (\ref{eq.imag}) are strict as soon as $\Real z\neq 0$.
\end{Lem}

\begin{proof}
Let $\mathcal{O}$ denote the open set $\mathcal{O}=\{  z\in \bbC^{+}\mid \Real(z) > 0 \}$. By virtue of H1, $\Imag f$ is an harmonic function on $\mathcal{O}$ that is continuous on $\operatorname{cl}\mathcal{O}$. H3 and H4 imply respectively that $f$ is real on the imaginary axis and  that $\Imag f(z)\geq 0$ on the positive real axis, thus we get that $\Imag f(z)\geq 0$ on the boundary $\partial \mathcal{O}$ of $\mathcal{O}$. Moreover, from H2 it follows that $\Imag f(z)\to 0 $ as $|z|\to \infty$ in $\operatorname{cl}\mathcal{O}$. All these conditions allow us to apply the maximum principle on the function $\Imag f$ in the unbounded domain $\mathcal{O}$ (see Corollary 4 p 246 of \cite{Dautray:1974:MAN}) which yields the inequality (\ref{eq.imag}) for $\Real(z)\geq 0$.
The inequality (\ref{eq.imag}) for $\Real(z)\leq 0$ is then deduced by using H3.

In the case where $f$ is not a constant function, by contradiction, if there exists a $z_0\in \bbC^{+}$ with a positive real part such that $\Imag f(z_0)=0$, then by the open mapping theorem the image by $f$ of an open ball $B(z_0,\delta)\subset \mathcal{O}$ is an open set of $\bbC$ which contains a real number $f(z_0)$ and therefore some points with a negative imaginary part. This contradicts (\ref{eq.imag}). Finally, by using H3, one obtains also that $\Imag f(z)<0$ for $\Real(z)<0$.
\end{proof}

To construct a Stieltjes function associated with $f$, we will follow the idea proposed by the authors of \cite{Milton:1997:FFR}.
For that purpose, we define the complex root by
\begin{equation}\label{eq.rac}
\sqrt{z}=|z|^{\frac{1}{2}}\,e^{\ii\arg z/2} \ \mbox{ if } \ \arg{z} \in (0,2\pi)
\end{equation}
and extend it on the branch cut $\bbR^{+}$ by its limit from the upper-half plane, in other words the square root of positive real number $x$ is given by $\sqrt{x}=|x|^{\frac{1}{2}}$.

\begin{Thm}\label{th.Stielt}
If $f$ satisfies the hypotheses  H1-4, then the function $u$ defined by
\begin{equation}\label{eq.funcstielt}
u(z):=f(\sqrt{-z}) , \ \forall z \in \bbC
\end{equation}
is a Stieltjes function which is positive on $\bbR^{+*}$.
\end{Thm}

\begin{proof}
The definition of the complex square root and the hypothesis H1 directly imply that $u$ is analytic on $\bbC\setminus \bbR^{-}$. Moreover, using the property H3 and the Lemma \ref{lem.partimag}, we get that $$u(\bbC^{+})=f(\{z\in \bbC^{+}\mid \Real(z)<0 \})\subset \operatorname{cl}\bbC^{-},$$
where $\bbC^{-}$ denotes the set $\bbC^{-}=\{ z \in \bbC \mid \Imag(z) < 0\}$.
To prove that $u$ is a Stieltjes function  positive on $\bbR^{+*}$, it just remains to show that $u(x)>0$ for $x>0$. By using H3, we immediately get that $u(x)=f(\ii \, x^{\frac{1}{2}}) \in \bbR$, for $x>0$. Then, the positivity of $u(x)$ follows from the decreasing nature of the real function $y\mapsto f(\ii \,y), \, y\in \bbR^+$ which implies, by virtue of H2, that $f(i \,y) \geq \lim_{y \to \infty} f(\ii y)=f_{\infty} >0$. This decreasing property is an immediate consequence of the Cauchy--Riemann relations written on the positive imaginary axis: $\partial_y \Real f (0, y)=-\partial_x \Imag f(0, y)$ and the fact that $\partial_x \Imag f (0, y)\geq 0$ by Lemma \ref{lem.partimag}.
\end{proof}
Figure \ref{fig.imagpart} sums up the effect of the square root mapping applied to the function $f$ to convert it into the Stieltjes function $u$, defined by (\ref{eq.funcstielt}).

\begin{figure}[!t]
\centering
 \includegraphics[width=0.9\textwidth]{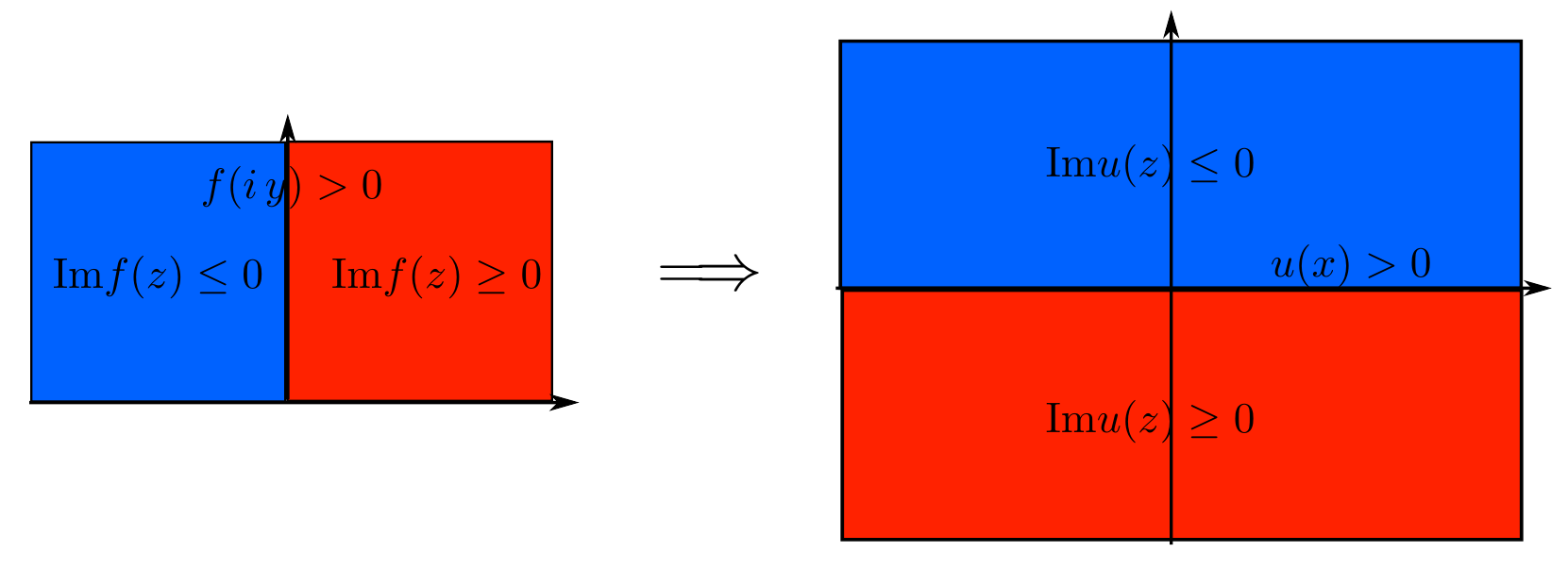}
 \caption{Sign of the imaginary part of the functions $f$ (left) and $u$ (right).}
 \label{fig.imagpart}
\end{figure}

\begin{Cor}\label{cor.herg}
The function $v$ defined by 
\begin{equation}\label{eq.funcherglotz}
v(z):=z\, u(-z)=z f(\sqrt{z}), \ \forall z\in \bbC
\end{equation}
is a Herglotz function, which is analytic on $\bbC\setminus{\bbR^{+}}$ and negative on $\bbR^{-*}$. Moreover, in its representation given by Theorem \ref{thm.Herg}, the measure $\mm$ is supported in $\bbR^{+}$ and $\alpha$ is equal to $f_{\infty}$. 
\end{Cor}

\begin{proof}
The following proof is partially inspired from \cite{Berg:2008:SPBS}.
First, one notices from definition (\ref{eq.funcherglotz}) and Theorem \ref{th.Stielt} that $v$ is analytic on $z\in \bbC\setminus{\bbR^{+}}$ and negative on $\bbR^{-*}$.
Then, as $u$ is defined by (\ref{eq.funcstielt}) is a Stieltjes function which tends to $f_{\infty}$ when $|z|\to \infty$, by the representation Theorem \ref{thm.Stielt}, $u$ can be expressed as
$$
u(z)=f_{\infty}+\int_{\bbR^{+}} \frac{\md \nu( \xi)}{\xi+z} \, , \ \forall z\in \bbC\setminus \bbR^{-},
$$
with $\nu$ a positive regular Borel measure on $\bbR^{+}$ such that $\int_{\bbR^{+}} \
 \md \nu (\xi)/(1+\xi)$ is finite.  
Thus, the function $v$ defined by  (\ref{eq.funcherglotz}) is given by 
\begin{eqnarray*}
v(z)&=& f_{\infty}\,z+ \int_{\bbR_{+}}\frac{z \md \nu( \xi)}{\xi-z}, \ \forall z \in \bbC\setminus{\bbR^{+}},
\end{eqnarray*}
and therefore
$$
\Imag v(z)=f_{\infty} \, \Imag(z)+\int_{\bbR_{+}} \frac{ \xi  \Imag(z)}{|\xi-z|^2} \, \md \nu( \xi)>0, \, \mbox{ when } \Imag(z)>0.
$$
Hence, one concludes that $v$ is a Herglotz function. Furthermore, from the definition (\ref{eq.funcherglotz}) of $v$ and the hypothesis H2, one gets immediately that its coefficient $\alpha$ in the representation Theorem \ref{thm.Stielt} is equal to $f_{\infty}$. Finally,  as $v$ is analytic on $z\in \bbC\setminus{\bbR^{+}}$ and negative on $\bbR^{-*}$, one deduces from (\ref{eq.mesHerg}) that the support of the measure $\mm$ associated with $v$ is included in $\bbR^{+}$.
\end{proof}

\begin{Rem}\label{Rem.DrudeLorentzmodel}
The assumption H1 supposes that $f$ can be continuously extended from the upper-half plane to the real line and implies in particular that $f$ admits no poles on the real axis or equivalently (see \cite{Gesztesy:2000:MVH}) that the measure $\nu$ associated with the Herglotz function $v$ has no punctual part. Indeed, we can relax this hypothesis by considering functions $f$ of the form:
\begin{equation}\label {eq.generalform}
f(z)=f_c(z)+f_p(z)   \mbox{ with } f_p(z)=-\sum_{n=1}^{N}\displaystyle \frac{ A_n}{z^2-\xi_n}  \mbox{ for } N\in \mathbb{N}, \ A_1,\cdots, A_N>0 \mbox{ and } \xi_1,\cdots, \xi_N \geq 0, 
\end{equation}
where $f_c$ satisfies the hypotheses H1-4. Then, it is straightforward to check that Theorem \ref{eq.funcherglotz} and Corollary \ref{cor.herg} still hold (except that  the definitions (\ref{eq.funcstielt}) and (\ref{eq.funcherglotz}) of $u$ and $v$ do not hold on the poles of $f$). In electromagnetism, the function $f$  can be seen as the dielectric permittivity $\Gve$ or the magnetic permeability $\mu$ as a function of the frequency $z=\Go$. In this context, functions $f=f_{\infty}+f_p$ correspond to the constitutive laws of non-dissipative generalized Lorentz models for which $\Gve$ or $\mu$ are rational functions of the frequency with real coefficients (see \cite{Veselago:1967:ESS,Tip:2004:LDD,Cassier:2016:STM}).
\end{Rem}

\begin{Rem}
In the literature \cite{Cessenat:1996:MME, Bernland:2011:SRC,Welters:2014:SLL}, one finds also another Herglotz function constructed from functions $f$ satisfying the hypothesis H1-4, namely 
\begin{equation}
\label{eq.otherhergfunc}\tilde{v}(z)=z \, f(z), \ \forall z \in \operatorname{cl}\bbC^{+}. 
\end{equation}
Indeed, to prove that the imaginary part of $\tilde{v}$ is non-negative for $z\in\bbC^{+}$, one follows the same arguments as in the proof of Lemma \ref{lem.partimag} by applying the maximum principle to the function $\operatorname{Im}\big(z \, (f-f_{\infty})\big)$ on $\operatorname{cl}\bbC^{+}$. Nevertheless, instead of H2, this requires a more stringent decreasing assumption at infinity: $f(z) =f_{\infty} +o(1/z), \mbox{ with } f_{\infty}>0$ when $\left|z\right|\to \infty$ in $\operatorname{cl} \bbC^{+}$. 

Unlike $v$, $\tilde{v}$ does not derive from a Stieltjes function, thus it does not satisfy the additional properties that it has an analytic extension in $\bbC\setminus{\bbR^{+}}$ which is negative on $\bbR^{-*}$. Therefore, the measure $\mm$ associated with $\tilde{v}$ in Theorem \ref{thm.Herg} is not necessarily supported in $\bbR^{+}$. Nevertheless, $\tilde{v}$ has the advantage to satisfy the additional relation: 
\begin{equation}\label{eq.sym}
h(z)=-\overline{h(-\overline{z})}, \ \forall z\in \operatorname{cl}\bbC^{+}
\end{equation}
(which can be deduced from H3). One will see in the following that using $\tilde{v}$ instead of $v$ will lead to slightly different bounds on the function $f$.
\end{Rem}

\subsection{General bounds on the function $f$}

Our aim is now to derive bounds on a function $f$ which satisfies the hypotheses H1-4 on a finite interval $[x_-,x_+] \subset \bbR^{+*}$. The key step is to use the analytic properties of its associated Herglotz function $v$ defined by Corollary \ref{cor.herg} which relies on the existence of the Stieltjes function $u$ of Theorem \ref{th.Stielt}. To this end, we follow the approach of \cite{Gustafsson:2010:SRP,Bernland:2011:SRC} by using the sum rules integral identities established in \cite{Bernland:2011:SRC}, recalled here in Proposition \ref{pro.sumrules}. Our resulting bounds generalize the ones developed in \cite{Gustafsson:2010:SRP,Bernland:2011:SRC}. Moreover, they are optimal in the sense that they maximize the sum rules \eqref{eq.sumrules} over the finite interval $[x_-,x_+] \subset \bbR^{+,*}$ in the sense of Theorem \ref{thm.mesopt}. 

Let $\GD>0$, we denote by $h_{\mm}$ the Herglotz function defined by:
\begin{equation}\label{eq.Herg}
h_{\mm}(z)=\int_{-\GD}^{\GD} \frac{\md  \mm( \xi) }{\xi-z} , \ \forall z\in \bbC^{+},
\end{equation}
 where $\mm\in \CM_{\GD}$. Here $\CM_{\GD}$ stands for the set finite regular positive Borel measure $\mm$ whose support is included in the interval $[-\GD,\GD]$ and whose total mass is normalized to $1$, in other words: $ \mm(\bbR)=\mm([-\GD,\GD])=1$, for all $\mm\in \CM_{\GD}$.
 
Our goal is to derive bounds on $f$ by using the sum rule \eqref{eq.sumrules} on the function $v_{\mm}$:
$$
v_{\mm}(z)=h_{\mm}(v(z)) \ \mbox{ on } \bbC^{+}
$$
where $v$ is the Herglotz function defined via $f$ in Corollary \ref{cor.herg}.
As $v$ is not constant, one first notices that $v_{\mm}$ is a Herglotz function as it is a composition of two Herglotz functions (see \cite{Berg:2008:SPBS}). To apply the sum rules, we need the asymptotic behavior of $v_{\mm}$ near zero and infinity. It is the purpose of the following lemma.

\begin{Lem}
For any $\theta \in (0,\frac{\pi}{2})$, the Herglotz function $h_{\mm}$ satisfies the following asymptotics in the Stolz domain $D_{\theta}$:
\begin{equation}\label{eq.asymphdelta}
h_{\mm}(z)=-\frac{\mm(\{0\})}{z}+o\left(\frac{1}{z}\right) \ \mbox{ as } |z| \to 0 \ \mbox{ and } \ h_{\mm}(z)=-\frac{1}{z}+o\left(\frac{1}{z}\right) \, \mbox{ as } \, |z| \to +\infty,
\end{equation}
which imply that in $D_{\theta} $:
\begin{equation}\label{eq.asympvdelta}
v_{\mm}(z)=  -\frac{\mm(\{0\})}{f(0)z}+o\left(\frac{1}{z}\right) \ \mbox{ as } |z| \to 0   \ \mbox{ and } \ v_{\mm}(z)=-\frac{1}{f_{\infty} \, z}+o\left(\frac{1}{z}\right) \, \mbox{ as } \, |z| \to +\infty.
\end{equation}
\end{Lem}

\begin{proof}
The asymptotic behavior at $z=0$ of $h_{\mm}$ follows from the relation (\ref{eq.asymherglotzfunction}) which is proved in \cite{Bernland:2010:SRC} by using Lebesgue's dominated convergence theorem.  To show the asymptotics (\ref{eq.asymphdelta}) of $h$ at $z=\infty$, one gets first, using the relation (\ref{eq.Herg}), that:
$$
z h_{\mm}(z)=   \int_{-\GD}^{\GD} \frac{ z\md  \mm( \xi) }{\xi-z}= -\mm([-\GD,\GD])+  \int_{-\GD}^{\GD} \frac{ \xi\md  \mm( \xi) }{\xi-z},
$$
where $\mm([-\GD,\GD])=1$, by hypothesis. One finally concludes by proving that the integral of the right hand side in the latter expression tends to $0$ as $|z| \to +\infty$. This is a consequence of Lebesgue's dominated convergence theorem where a domination condition on the integrand is given by
$$
\displaystyle \left |\frac{ \xi}{\xi-z}\right |\leq \frac{1}{\sin(\theta)}, \ \mbox{ since } \, |\xi-z|\geq |\xi| \sin(\theta), \ \forall \xi \in [-\GD,\GD] \mbox{ and } \forall  z\in D_{\theta}.
$$
The asymptotics (\ref{eq.asympvdelta}) follows immediately by composition from the asymptotics (\ref{eq.asymphdelta}) and the hypotheses H1 and H2 which imply respectively that in $D_{\theta}$: $v(z)=f(0)z+o(z), \mbox{ as } \ |z| \to 0$  and $v(z)=f_{\infty}\, z+o(z), \mbox{ as } \ |z| \to +\infty$ (one notices that the first asymptotic  formula in (\ref{eq.asympvdelta}) is well-defined. Indeed $f(0)$ is positive since we already showed (see proof of Theorem \ref{th.Stielt}) that the function $f(z)$ is real and decreasing along the imaginary axis, thus $f(0) \geq f_{\infty} >0$).
\end{proof}

One can now use the sum rules  \eqref{eq.sumrules} on the function $v_{\mm}$ over the finite frequency band $[x_-,x_+]$ to get the following inequality:
\begin{equation}\label{eq.sumrulesmu}
\lim_{y\to 0^{+}} \frac{1}{\pi} \int_{x_-}^{x_+} \Imag v_{\mm}(x+\ii y) \, \md x \leq \displaystyle \frac{1}{f_{\infty} } -\frac{\mm(\{0\})}{f(0)}  \leq \frac{1}{f_{\infty}},
\end{equation}
where the right inequality in the latter expression is justified by the fact that $\mm$ is a positive measure and that $f(0)$ is also positive. 

The following theorem expresses that if one wants to maximize the sum rules (\ref{eq.sumrulesmu}) on the set of measures $\CM_{\GD}$, it is sufficient to use Dirac measures: $\mm=\delta_{\xi}$ for points $\xi \in  [-\GD,+\GD]$.

\begin{Thm}\label{thm.mesopt}
Let $\GD$ be a positive real number and $[x_-,x_+]$ be a finite frequency band included in $\bbR^{+*}$, then one has
\begin{equation}\label{eq.maxsumrule}
\sup_{\mm \in \CM_{\GD}}  \frac{1}{\pi} \lim_{y\to 0^{+}}  \int_{x_-}^{x_+} \Imag v_{\mm}(x+\ii y) \, \md x=\sup_{\xi \in [-\GD,+\GD]} \frac{1}{\pi} \lim_{y \to 0^{+}}  \int_{x_-}^{x_+} \Imag v_{\delta_{\xi}}(x+\ii y) \, \md x.
\end{equation}
\end{Thm}

\begin{proof}
Let $\mm \in \CM_{\GD}$, one denotes by $\nu_{\mm}$ the measure associated with the Herglotz function: $v_{\mm}$ by the representation Theorem \ref{thm.Herg}. Thus, by virtue of the relation (\ref{eq.mesHerg}) which defines the measure of a Herglotz function, one has:
\begin{equation}\label{eq.hergmesproof}
 \frac{1}{\pi} \lim_{y\to 0^{+}}  \int_{x_-}^{x_+} \Imag v_{\mm}(x+\ii y) \, \md x = \frac{\nu_{\mm}\big((x_-,x_+)\big)+\nu_{\mm}([x_-,x_+])}{2}.
\end{equation}
One wants now to connect the measure $\nu_{\mm}$ of $v_{\mm}=h_{\mm}\circ v$ to the measure $\mm$ of the Herglotz function $h_{\mm}$. In \cite{Christodoulides:2004:GVD}, the authors provide an expression of the measure $\nu_{\mm}$ in terms of the measures $\mm$ and the measure $\nu_{\delta_{\xi}}$ associated with the Herglotz function $v_{\delta_\xi}=h_{\delta_{\xi}}\circ v$. They prove that for any Borelian sets $B$, $\nu_{\mm}(B)$ is given by
$$
\nu_{\mm}(B)=\int_{-\GD}^{\GD} \nu_{\delta_\xi}(B)\md \mm(\xi).
$$ 
Thus, applying this last relation to $B=(x_-,x_+)$ and $B=[x_-,x_+]$ in the equation (\ref{eq.hergmesproof}) yields:
\begin{eqnarray*}
 \frac{1}{\pi} \lim_{y\to 0^{+}}  \int_{x_-}^{x_+} \Imag v_{\mm}(x+\ii y) \, \md x &=&\int_{-\GD}^{\GD} \frac{1}{2}\left[ \nu_{\delta_\xi}((x_-,x_+)) +\nu_{\delta_\xi}([x_-,x_+])\right] \md \mm(\xi),\\
 &\leq&  \sup_{\xi \in [-\GD,+\GD]}  \left( \frac{1}{2}\left[ \nu_{\delta_\xi}((x_-,x_+)) +\nu_{\delta_\xi}([x_-,x_+])\right]  \right) \mm([-\GD,+\GD]).
\end{eqnarray*}
Using now the fact that $\mm \big([-\GD,\GD]\big)=1$ and the relation (\ref{eq.mesHerg}) which characterizes the measure $\nu_{\delta_\xi}$ of the Herglotz function $v_{\delta_\xi}$ lead us to:
$$
 \frac{1}{\pi} \lim_{y\to 0^{+}}  \int_{x_-}^{x_+} \Imag v_{\mm}(x+\ii y) \, \md x \leq  \sup_{\xi \in [-\GD,+\GD]}   \frac{1}{\pi} \lim_{y \to 0^{+}}  \int_{x_-}^{x_+} \Imag v_{\delta_{\xi}}(x+\ii y) \, \md x.
$$
By taking the supremum on $\CM_{\GD}$, one shows one side of the equality \eqref{eq.maxsumrule}.
As the reverse inequality of \eqref{eq.maxsumrule} is straightforward, this concludes the proof.
\end{proof}
The last theorem shows that  for any $\Delta \in \bbR$, the family of Dirac measures $(\delta_{\xi})_{\xi \in \bbR}$ maximizes the sum rule (\ref{eq.sumrulesmu}) on the set of probability measures $\CM_{\Delta}$. For such measures, the inequality (\ref{eq.sumrulesmu}) can be rewritten as:
\begin{equation}\label{eq.sumruledelta}
 \lim_{y \to 0^{+}}  \int_{x_-}^{x_+} \Imag v_{\delta_{\xi}}(x+\ii y) \, \md x=\lim_{y\to 0^{+}}  \int_{x_-}^{x_+} \Imag \left( \displaystyle \frac{1}{\xi-v(x +\ii y)}\right) \, \md x \leq  \frac{\pi}{f_{\infty}}, \  \forall \xi \in \bbR \,.
\end{equation}

\subsection{The case of a transparency window}
By using the family of punctual measures $(\delta_{\xi})_{\xi \in \bbR}$, we will now derive an explicit bound on the function $f$ on a interval $[\Go_-,\Go_+]\subset \bbR^{+*}$ under the assumption that this interval is a transparency window.
In other words, one supposes that $f$ is real on $[\Go_-,\Go_+]$. In physics, (like in electromagnetism when for instance $f=\Gve$ or $f=\mu$), this hypothesis amounts neglecting the absorption of the material in the frequency band $[\Go_-,\Go_+]$. 

In this case, one gets immediately that the Herglotz function $v$ is real on $[x_-,x_+]=[\Go_-^2, \Go_+^2]$. 
Thus, one can extend $v$ analytically through the interval $(x_-,x_+)$ by using Schwarz's reflection principle by posing $$v_e(z)=v(z) \ \mbox{ on } \operatorname{cl} \bbC^{+} \  v_e(z)=\overline{v(\overline{z})} \mbox{ on } \bbC^{-}$$ 
and it is straightforward to check (thanks to H3) that $v_e$ coincides with the definition (\ref{eq.funcherglotz}) of $v$ on the domain $D=\bbC \setminus ([0,x_-]\cup [x_+,+\infty])$. Hence, the function $v$ is analytic on $D$. With our approach, one recovers in the next proposition a bound similar to the ones derived in \cite{Milton:1997:FFR,Yaghjian:2006:PWS}. This bound correlates the value of two points  of the function $f$ within the considered interval. A generalization of such bounds to an arbitrary number of points of correlation is done in \cite{Milton:1997:FFR}.
\begin{Pro}\label{Pro.boundtransp}
In the transparency window $[x_-,x_+] =[\Go_-^2, \Go_+^2]$, the function $v$ satisfies
\begin{equation}\label{eq.ineqv}
f_{\infty} (x-x_0) \leq v(x)-v(x_0), \  \forall x, x_0 \in [x_-,x_+] \mbox{ such that } \ x_0\leq x,
\end{equation}
which yields the following bound on $f$:
\begin{equation}\label{eq.boundtranspf}
\Go_0^2 (f(\Go_0)-f_{\infty}) \leq \Go^2  (f(\Go)-f_{\infty}), \  \forall \Go, \Go_0 \in [\Go_-,\Go_+] \mbox{ such that } \ \Go_0\leq \Go.
\end{equation}
\end{Pro}

\begin{proof}
Let $x_0 \in (x_-,x_+)$. One defines $\xi$ by $\xi=v(x_0)\in \bbR$. Hence, the Herglotz function $v_{\delta_{\xi}}=(\xi-v)^{-1}$ has a pole at $z=x_0$. As any real pole of a Herglotz functions is of multiplicity one (see \cite{Gesztesy:2000:MVH}), this implies in particular that the derivative $v^{\prime}(x_0)\neq 0$. Moreover, as $\xi-v$ is an analytic function which is not constant on $D$, therefore the pole $x_0$ is isolated. Thus, there exists a closed interval: $[\tilde{x}_-,\tilde{x}_+] \subset (x_-,x_+)$ containing $x_0$ such that $x_0$ is the only singular point of the function $v_{\delta_\xi}$  on  $[\tilde{x}_-,\tilde{x}_+]$. Hence, one can rewrite $v_{\delta_\xi}$ as:
$$
v_{\delta_\xi}(z)=\frac{g(z)}{(z-x_0)}\  \mbox{ where $g$ is analytic and real on $ [\tilde{x}_-,\tilde{x}_+]$ and } g(x_0)=\frac{-1}{v^{\prime}(x_0)}.
$$
Using this last property on $v_{\delta_\xi}$, one can evaluate the limit in the left hand side of \eqref{eq.sumruledelta}:
\begin{eqnarray}\label{eq.limitplemej}
 && \lim_{y \to 0^{+}}  \int_{\tilde{x}_-}^{\tilde{x}^{+}} \Imag v_{\delta_{\xi}}(x+\ii y) \, \md x   \nonumber\\
= & & \lim_{y\to 0^{+}}  \int_{\tilde{x}_-}^{\tilde{x}^{+}} \Imag \left( \frac{g(x+\ii y)-g(x_0)}{x+\ii y-x_0}\right) \md x+\lim_{y  \to 0^{+}}   \int_{\tilde{x}_-}^{\tilde{x}^{+}} \Imag \left( \frac{-g(x_0)}{x-(x_0+ \ii y)}\right) \md x.
\end{eqnarray}
 Indeed, as a consequence of Lebesgue's dominated convergence theorem, the first limit of (\ref{eq.limitplemej}) is $0$ and by  applying the Sokhotski-Plemelj formula (see \cite{Henrici:1993:ACC}) to evaluate the second limit of (\ref{eq.limitplemej}), one gets:
$$
 \lim_{y \to 0^{+}}  \int_{\tilde{x}_-}^{\tilde{x}^{+}} \Imag v_{\delta_{\xi}}(x+\ii y) \, \md x=-\pi g(x_0)=\frac{\pi}{v^{\prime}(x_0)},
$$
By using \eqref{eq.sumruledelta}, this leads to:
$$
f_{\infty}\leq v^{\prime}(x_0), \ \forall x_0 \in (x_-,x_+).
$$
Integrating this latter relation leads to inequality (\ref{eq.ineqv}) on $(x_-,x_+)$, which extends to the closed interval $[x_-,x_+]$ by using the continuity of $v$ at $x_{\pm}$. One finally derives inequality (\ref{eq.boundtranspf}) from (\ref{eq.ineqv}) by using the definition (\ref{eq.funcherglotz}) of $v$ and the changes of variables: $x=\Go^2$ and $x_0=\Go_0^2$.
 \end{proof}

\subsubsection*{Link with the Kramers--Kronig relations}
For the case of a transparency window: $[\Go_-,\Go_+]$, we want now to emphasize that the bound obtained in the proposition \ref{Pro.boundtransp} can be also derived by applying the Kramers--Kronig relations to the function $f$:
 \begin{equation}\label{eq.Kramerkronig}
 \Real f (\Go) =f_{\infty}+\frac{2}{\pi}\CP\int_0^\infty\frac{\Go' \Imag f (\Go')}{(\Go')^2-\Go^2}\,d\Go'
\end{equation}
where $\CP$ denotes the Cauchy principal value of the integral.   
 In electromagnetism, these relations, satisfied by the permittivity  $\Gve$ and the permeability $\mu$ (see \cite{Nussenzveig:1972:CDR,Jackson:1999:CE}), characterize the dispersion of a passive material by correlating the real part and the imaginary part of $\Gve$ and $\mu$ by nonlocal
integral relations.

Mathematically, to derive the Kramers--Kronig relations pointwise at a frequency $\Go$, one supposes classically in addition to H1-4 that $\Go' \to (f-f_{\infty})/\Go'$ is an integrable function at the vicinity of $\pm \infty$ and that $f$ is H\"{o}lder continuous at $\Go$. These two last conditions (see \cite{Henrici:1993:ACC,Pecseli:2000:FPS}) ensure the existence of the Cauchy principal value in (\ref{eq.Kramerkronig}). We point out that in the literature, one can find other mathematical hypotheses such as $f$ belongs to to Hardy space  $H_2(\bbC^{+})$ (see Titchmarsh's theorem \cite{Nussenzveig:1972:CDR}) which ensure the existence of these relations for almost every real frequencies $\Go$.

Now using the fact $[\Go_-,\Go_+]$ is a transparency window, i. e. $ \Imag f (\Go)=0$ for all $\Go\in [\Go_-,\Go_+]$, one gets:
\begin{equation}\label{eq.kramerkronig}
f (\Go)=f_{\infty}+\frac{2}{\pi} \int_0^{\Go_-}\frac{\Go' \Imag f (\Go')}{(\Go')^2-\Go^2}\,d\Go'+\frac{2}{\pi}\int_{\Go^{+}}^\infty \frac{\Go' \Imag f (\Go')}{(\Go')^2-\Go^2}\,d\Go', \ \forall \Go\in (\Go_-,\Go_+),
\end{equation}
where the Cauchy principal value is not useful anymore in the latter expression since in both integrals the singular point does not belong to the domain of integration. Moreover, in a transparency window, the function $f$ can be analytically, by a Schwarz reflection principle, extended through the interval, thus the  H\"{o}lder regulartity is satisfied on $(\Go_-,\Go_+)$.

Applying the Kramers-Kronig relation (\ref{eq.kramerkronig}) to two frequencies $\Go$, $\Go_0 \in(\Go_-,\Go_+)$ satisfying $\Go_0\leq \Go$  yields

\begin{eqnarray*} \Go^2[f(\Go)-f_\infty]-\Go_0^2[f(\Go_0)-f_\infty] & = &
\frac{2}{\pi}\int_0^{\Go_-} \Go' \Imag f(\Go')\left[\frac{\Go^2}{(\Go')^2-\Go^2}-\frac{\Go_0^2}{(\Go')^2-\Go_0^2}\right]\,d\Go' \nonum
&\,&+\frac{2}{\pi}\int_{\Go_+}^\infty \Go' \Imag f(\Go')\left[\frac{\Go^2}{(\Go')^2-\Go^2}-\frac{\Go_0^2}{(\Go')^2-\Go_0^2}\right]\,d\Go' \nonum
& = &\frac{2}{\pi}\int_0^{\Go_-} \Go' \Imag f(\Go')\left[\frac{(\Go')^2(\Go^2-\Go_0^2)}{[(\Go')^2-\Go^2][(\Go')^2-\Go_0^2]}\right]\,d\Go' \nonum
&\,& +\frac{2}{\pi}\int_{\Go_+}^\infty \Go' \Imag f(\Go')\left[\frac{(\Go')^2(\Go^2-\Go_0^2)}{[(\Go')^2-\Go^2][(\Go')^2-\Go_0^2]}\right]\,d\Go' \nonum
&\geq & 0,
\end{eqnarray*}
where to obtain the last inequality we have used the fact that the ratio
\begin{equation*} \frac{(\Go^2-\Go_0^2)}{[(\Go')^2-\Go^2][(\Go')^2-\Go_0^2]} \geq 0
\end{equation*}
when either $\Go\geq\Go_0> \Go_-\geq \Go'>0$ or when $\Go'\geq \Go_+>\Go\geq\Go_0>0$ and the fact that H4 imposes that $\Go' \Imag f(\omega')$ is positive on $\bbR^{+}$. Thus, one obtains  again the bound (\ref{eq.boundtranspf}) on the open interval $(\omega_-,\omega_+)$. Finally, this bound can be extended  to the closure of this interval  by using the continuity of $f$ at $\omega_{\pm}$.

\subsection{The lossy case}\label{lossy-mat}
The bound (\ref{eq.boundtranspf}) is only valid if $\Imag f$ is exactly zero on $[\Go_-,\Go_+]$. When the loss of the material cannot be neglected in this frequency band, other bounds can be derived from the inequality (\ref{eq.sumrulesmu}). 
By choosing for instance the uniform measure of $\CM_{\Delta}$: $$\md \mm(\xi)= \frac{\bf{1}_{[-\GD,\GD]}(\xi)}{2 \GD}  \md \xi$$ 
for the Herglotz function $h_{\mm}$, one recovers the bounds derived in \cite{Gustafsson:2010:SRP}. More precisely, we get:
\begin{equation}\label{eq.Hergunif}
h_{\mm}(z)=\frac{1}{2\GD} \int_{-\GD}^{\GD}\frac{1}{\xi-z}\md \xi=\frac{1}{2\GD}\log \left( \frac{z-\GD}{z+\GD}\right), \ \forall z \in \bbC^{+},
\end{equation}
where the function $\log$ is defined with the same branch cut: $\bbR^{+}$ as the square root function (\ref{eq.rac}). 
As:
$$
 \frac{z-\GD}{z+\GD}=\frac{|z|^2-\GD^2+2\ii \GD \Imag(z)}{|z+\GD^2|}, \ \forall z \in \bbC^{+},
$$
one checks easily that $\Imag h_{\mm}(z)$ is bounded above by $\pi/(2\GD)$ and from below by: 
\begin{equation}\label{eq.boundimag}
 \Imag  h_{\mm}(z) \geq \frac{\pi}{4 \GD} H(\GD-|z|), \ \forall z \in \bbC^{+},
\end{equation}
where $H$ stands here for the Heaviside function. Moreover, in the limit $y  \to 0^+$, $\Imag[h(x+i y)]$ takes the value $\pi/(2\Delta)$ for
$|x|<\GD$ and $0$ for $|x|>\GD$.

Now, applying the relations (\ref{eq.Hergunif}) and (\ref{eq.sumrulesmu}), one gets:
\begin{equation}\label{eq.sumrulelog}
\lim_{y \to 0^{+}} \int_{x_-}^{x_+} \Imag v_{\mm}(x+\ii y)\, \md x =\frac{1}{2\GD}\lim_{y \to 0^{+}} \int_{x_-}^{x_+} \operatorname{arg}\left( \frac{v(x+\ii y)-\GD}{v(x+i\Ge)+\GD}\right) \md x\leq \frac{\pi}{f_{\infty}}.
\end{equation}
Hence, using the bound (\ref{eq.boundimag}), one gets we obtain a less stringent
but more transparent inequality:
$$
\lim_{y\to 0^{+}}  \int_{x_-}^{x_+}  H(\GD-|v(x+\ii y )|) \, \md x  \leq\frac{4\GD}{f_{\infty}}
$$
and using Lebesgue's Dominated convergence theorem to evaluate this limit (thanks to  the continuity assumption of $f$ on $[x_-,x_+]$) we get
\begin{equation}\label{eq.bound1losscase}
 \int_{x_-}^{x_+}H(\GD-|v(x)|)\, \md x
\leq\frac{4\GD}{f_{\infty}}.
 \end{equation}
 In a plot of $|v(x)|$ against $x$ the quantity on the left of (\ref{eq.bound1losscase}) represents the total length of the interval or intervals  of $x$, between $x_-$ and $x_+$, where $|v(x)|$ is less than $\GD$. Clearly the bound implies that this total length must shrink to zero as $\GD\to 0$. If we take
\begin{equation*} 
\GD=\max_{x\in[x_-,x_+]}|v(x)|
\end{equation*}
then the left hand side of (\ref{eq.bound1losscase}) equals $x_+ - x_-$ and 
\begin{equation*}
\frac{1}{4}(x_+-x_-) f_{\infty} \leq \max_{x\in[x-,x_+]}|v(x)|.
\end{equation*}
One finally gets immediately from this last inequality the following bound on the function $f$.

\begin{Pro}\label{prop.lossy}
Let $[\Go_-,\Go_+]\subset \bbR^{+*}$ then the function $f$ satisfies the following inequality:
\begin{equation}\label{eq.bound2losscase}
\frac{1}{4}(\Go_+^2-\Go_-^2) f_{\infty} \leq \max_{x\in[\Go_-,\Go_+]}|\omega^2 f(\Go)|.
\end{equation}
\end{Pro}

This last bound is essentially the same as the bound (1) derived in \cite{Gustafsson:2010:SRP}. 

\begin{Rem}
More precisely, one will recover exactly the bound (1) derived in \cite{Gustafsson:2010:SRP}, namely
\begin{equation*}
\frac{1}{2}(\Go_+-\Go_-) f_{\infty} \leq \max_{x\in[\Go_-,\Go_+]}|\omega f(\Go)| \,,
\end{equation*}
if one uses the Herglotz function $\tilde{v}$ defined by (\ref{eq.otherhergfunc}) instead of $v$ to define the function $v_{m}$. One points out that factor 1/2 instead of the factor 1/4 in (\ref{eq.bound2losscase}) comes from the relation (\ref{eq.sym}) satisfied by $\tilde{v}$ which allows one to rewrite the sum rule (\ref{eq.sumrules}) as
$$
\lim_{\eta \to 0^{+}}\lim_{y\to 0^{+}}\frac{2}{\pi} \int_{\eta<x<\eta^{-1}} \Imag h (x+\ii y)\, \md x=a_{-1}-b_{-1}.
$$
\end{Rem}

\begin{Rem}\label{Rem.weakhypcont}
Notice here that all the bounds derived in this section still hold 
for functions $f$ of the form (\ref{eq.generalform}) whose real poles do not belong to the interval $[\Go_-,\Go_+]$. In other words, it extends also to non-dissipative generalized Drude--Lorentz models whose resonances do not belong to the frequency range of interest. Thus, the hypothesis H1 which assumes that $f$ is continuous for all real frequencies can be relaxed.
\end{Rem}

%%%%%%%%%%%%%%%%%%%%%%%%%%%%%%%%%%%%%%%%%%%%%%%%%%%%%%%%%%%%%%%%%%%%%%%%%%

\section{Bounds on the polarizability tensor and quasi-static cloaking}\label{sec.boundcloak}

\subsection{Formulation of the problem}
The challenging problem we address in this section is the following: is it possible to construct a passive material to cloak a dielectric inclusion on a whole frequency band $[\Go_-,\Go_+]$?
Using the bounds derived in the first section, we will prove that it is not possible when one makes the quasi-static approximation of Maxwell's equations.

 Let $\CO$ be a bounded simply-connected dielectric inclusion with Lipschitz boundary and constant permittivity $\Gve \, \BI$  satisfying $\Gve>\Gve_{0}$. We assume here in particular that $\CO$ is made of a standard dielectric material for which one can neglect the dispersion, in other words the frequency dependence of $\Gve$, on the frequency range of interest $[\Go_-,\Go_+]$.
To make invisible $\CO$, one uses a passive cloak of any shape characterized by its dielectric tensor $\BGve(\Bx,\Go)$ which depends both on the spatial variable $\Bx$ and the frequency $\Go$. Thus, the cloak is  composed of an anisotropic, dispersive and heterogeneous material. The whole device: the dielectric inclusion and the cloak is assumed to fill a bounded open set $\Omega \subset \bbR^3$ of characteristic size $R_0$, in other words $\Omega \subset B(0,R_0)$ where $B(0,R_0)$ denotes the sphere of radius $R_0$ centered at the origin. Finally, one supposes that the rest of the space: $\bbR^3\setminus \GO$ has the same dielectric constant $\Gve_0 \, \BI$ as the vacuum. We emphasize that the cloak can surround the inclusion $\CO$ (like in the figure \ref{fig.med}) which is the case for many cloaking methods, but our results hold also for cloaking methods such as anomalous resonances \cite{Milton:2006:CEA} or complementary media \cite{Lai:2009:CMI} for which the inclusion can be outside the cloak.

\begin{figure}[!t]
\centering
 \includegraphics[width=0.6\textwidth]{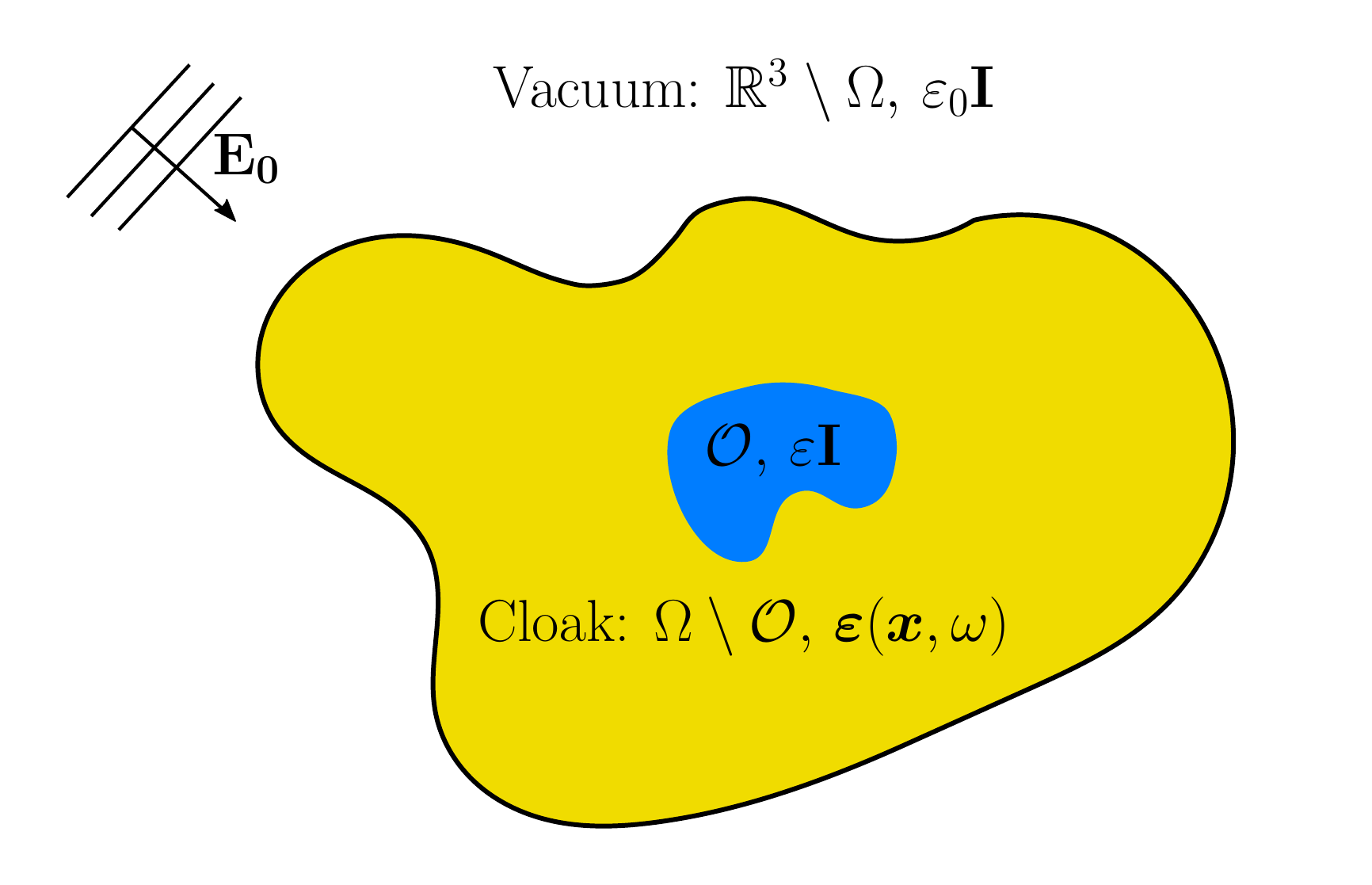}
 \caption{Description of the cloaking problem}
 \label{fig.med}
\end{figure}

For simplicity we assume that there is a plane incident wave on the device, with wavelength considerably larger than $R_0$, so that within the frequency range of interest
$\Go\in [\Go_-,\Go_+]$ we can use the quasi-static equations which amounts neglecting the term 
due to the time-derivatives of the electrical and magnetic inductions in the time-harmonic Maxwell equations.
Thus, it leads to a decoupling of  these equations. In this setting \cite{Jackson:1999:CE,Milton:2002:TC}, one can express the electrical field $\BE(\Bx,\Go)$ in terms of the gradient of some potential $V(\Bx,\Go)$, 
i.e. $\BE(\Bx,\Go)=-\Grad V(\Bx,\Go)$, and an incident plane wave corresponds to a 
uniform field $\BE_0\in \bbC^{3}$ at infinity so that the potential $\Grad V(\Bx,\Go)$ has to satisfy the following elliptic equation
\begin{equation}\label{eq.quasis}
\left\{ \begin{array}{ll}
 \Div \big( \BGve(\Bx,\Go)\Grad V(\Bx,\Go) \big) =0 \ {\rm on } \ \bbR^3, \\[5pt]
  V(\Bx,\Go)=-\BE_0\cdot\Bx+\CO(1/|\Bx|)\ {\rm as}\,\,|\Bx|\to\infty. 
\end{array} \right.
\end{equation}

In this context, the leading order correction to the uniform incident field $\BE_0$ at infinity is a dipolar field (see \cite{Ammari:2007:PMT,Jackson:1999:CE,Kang:2008:SPS,Milton:2002:TC}), so that the potential $V$ has the asymptotic expansion
\begin{equation} \label{eq.asymp}V(\Bx,\Go)=-\BE_0\cdot\Bx+\frac{\Bp(\Go)\cdot\Bx}{4\pi \Gve_0 |\Bx|^3}+\CO\Big(\frac{1}{|\Bx|^3}\Big)\ {\rm as}\,\,|\Bx|\to\infty. 
\end{equation}
where the induced dipole moment $\Bp(\Go)\in\bbC^3$ is linearly related to the applied field $\BE_0$ and this linear relation:
\begin{equation}\label{eq.polatens}\Bp(\Go)=\BGa(\Go) \BE_0
\end{equation}
defines the polarizability tensor $\BGa(\Go)$ (also called the Polya--Szego tensor), which is a $3\times 3$ complex matrix. We point out that $\BGa$ is  a function of the frequency in the case of dispersive media.
 $\BGa(\Go)$ defines the leading term of the far field of the scattered wave generated by the whole device $\GO$.
Hence, one says that the dielectric inclusion $\CO$ is cloaked at a sufficient large distance at a frequency $\Go \in [\Go_-,\Go_+]$, if  the polarizability tensor $\BGa(\Go)$ vanishes at $\Go$.

We emphasize that the equations (\ref{eq.quasis})-(\ref{eq.polatens}) which define the polarizability tensor $\BGa(\Go)$ are physically relevant only in the frequency interval of interest $[\omega_-,\omega_+]$ where the quasi-static
approximation is valid. 
Nevertheless, as the dielectric tensor of the cloak $\BGve(\Bx,\Go)$ is defined, by the constitutive laws, for all frequencies $\omega$ in the closure of the upper-half plane $\operatorname{cl}\bbC^{+}$, one can study mathematically these equations for $\Go \in \operatorname{cl}\bbC^{+}$. To define their extension to $\Go \in \operatorname{cl}\bbC^{+}$, we set also that within the dielectric inclusion $\CO$, the permittivity $\BGve(\cdot,\Go)$ is constant and equal to $\Gve \, \BI$ for all $\Go \in \operatorname{cl}\bbC^{+}$. This definition is also physically relevant only in the frequency band $[\omega_-,\omega_+]$ where $\CO$ is assumed to be a non-dispersive dielectric (since the only material which behaves as a non-dispersive media at all frequencies is the vacuum). Outside the cloaking device $\Omega$, where the dielectric behaves as the vacuum, one extends $\BGve(\cdot,\Go)$  by $\Gve_0 \, \BI$ for all $\Go \in \operatorname{cl}\bbC^{+}$.
To sum up, the extension of the equations (\ref{eq.quasis})-(\ref{eq.polatens}) to $\Go \in \operatorname{cl}\bbC^{+}$ is performed to derive quantitative bounds on $\BGa(\Go)$  which have  a physical meaning only in the frequency band of interest $[\omega_-,\omega_+]$.

We have now to specify what we mean by a passive cloak. In the following, we equip the space of complex $3\times 3$ matrices with the induced $l^2$ norm. A passive cloak is defined as a material which satisfies the following assumptions:
\begin{itemize}
\item $\tilde{\mathrm{H}}1$: for a. e. $\Bx \in \Omega\setminus \CO, \,\BGve(\Bx,\cdot)$ is analytic on $\bbC^{+}$ and continuous on $\operatorname{cl} \bbC^{+}$,
\item $\tilde{\mathrm{H}}2$: for a. e. $\Bx \in \Omega\setminus \CO,$  $\ \BGve(\Bx,\Go)\to \Gve_0 \BI \,  \mbox{ as } |\Go| \to \infty \mbox{ in} \operatorname{cl} \bbC^{+}$,
\item $\tilde{\mathrm{H}}3$: for a .e. $\Bx \in \Omega\setminus \CO, \, \forall \Go \in \operatorname{cl}\bbC^{+},\ \BGve(\Bx,-\overline{\Go})=\overline{\BGve(\Bx,\Go)}$,
\item $\tilde{\mathrm{H}}4$: for a. e. $\Bx \in \Omega\setminus \CO, \forall \Go \in \bbR^{+}, \ \operatorname{Im} \BGve(\Bx,\Go)\geq 0$  (passivity),
\item $\tilde{\mathrm{H}}5$: for a. e. $\Bx \in \Omega\setminus \CO, \forall \Go \in\operatorname{cl}  \bbC^{+}, \  \BGve(\Bx,\Go)^{\top}=\BGve(\Bx,\Go)$, where $\top$ stands for the transpose operation (reciprocity principle).
\end{itemize}
Assumptions $\tilde{\mathrm{H}}1-4$ correspond to hypotheses $\operatorname{H}1-4$ (given in section \ref{First part}) but are expressed in the more general case of anisotropic and heterogeneous passive materials \cite{Cessenat:1996:MME,Milton:2002:TC,Welters:2014:SLL}. The assumption $\tilde{\mathrm{H}}5$ is classical. Physically it  means that the cloak satisfies a reciprocity principle.
It is shared by most of the electromagnetic media, but it can be violated in some particular cases as for gyroscopic media or in the presence of Hall effect or magnetic-optical effect \cite{Landau:1984:ECM}. Most of our bounds still hold when the reciprocity principle $\tilde{\mathrm{H}}5$ is broken, therefore in the following we specify the results for which this additional hypothesis is required.

In the following, one extends the definition of $\BGve$ to $\Go=\infty$ by posing:
$$
\BGve(\infty,\Bx)=\Gve \BI  \  \mbox{ for a. e. } \Bx \mbox{ in } \CO \ \mbox{ and } \ \BGve(\infty,\Bx)=\Gve_0  \BI  \ \mbox{ for a. e. }  \Bx  \in \bbR^3\setminus \CO,
$$
so that $\BGve(\Bx,\cdot)$ is continuous on $\operatorname{cl}\bbC^+\cup\{\infty\}$ for a. e $\Bx\in \bbR^3$ (by hypotheses $\tilde{\mathrm{H}}1$ and $\tilde{\mathrm{H}}2$).

In this context, the broadband passive problem can be rephrased as follows: is it possible to construct a passive cloak, in other words, a material satisfying the five hypotheses $\tilde{\mathrm{H}}1-5$ in $\GO\setminus \CO$ such that the polarizability tensor $\BGa(\Go)$ associated with the whole device $\GO$ vanishes on the whole frequency band $[\Go_-,\Go_+]$? We will answer negatively to this question and derive quantitative bounds on the function $\BGa$ over this frequency range.

\subsection{Analyticity of the polarizability tensor}\label{sub.analypola}

To derive fundamental limits on the cloaking effect over the frequency band $[\Go_-,\Go_+]$, we want to apply the bounds derived in the section \ref{sec-boundanalycomp} to the polarizability tensor $\BGa$ or more precisely to the scalar function 
\begin{equation}\label {eq.falpha}
f(\Go)=\BGa(\Go) \BE_0\cdot \overline{\BE_0}.
\end{equation}
Hence, the first step is to prove that if the dielectric tensor $\BGve$ satisfies $\tilde{\mathrm{H}}1-5$ then the function $f$ satisfies the hypotheses $\mathrm{H}1-4$ that we used to derive these bounds.

To this aim, we first recall in this subsection why equations (\ref{eq.quasis}) are well-posed, in other words why they admit a unique solution $V(\cdot,\Go)$ in a classical functional framework. Moreover, we show that $V(\cdot,\Go)$ depends analytically on the frequency $\Go$ on  $\bbC^{+}$ and continuously on $\operatorname{cl}\bbC^{+} \cup \{\infty \}$. Then, we use this result to prove that the function $f$, defined by (\ref{eq.falpha}), shares the same regularity and satisfies the assumptions $\operatorname{H}1-4$. Finally, this allows us (using the results of subsection \ref{sub-secconstrHerg}) to construct a Herglotz function $v$ associated with $f$.

In this perspective, we seek the potential $V(\cdot,\Go)$, that is the solution of (\ref{eq.quasis}), in the form:
\begin{equation}\label{eq.pot}
V(\Bx,\Go)= -\BE_0\cdot \Bx+V_s(\Bx,\Go)
\end{equation}
where $V_s(\cdot,\Go)$ denotes the scattered potential due to the reflection of the uniform field $\BE_0$ on the device $\Omega$. Hence by (\ref{eq.quasis}), $V_s$ satisfies
\begin{numcases}{}
 \Div(\BGve(\Bx,\Go)\Grad V_s(\Bx,\Go)) =\Div\big((\BGve(\Bx,\Go)-\Gve_0 \BI)\, \BE_0\big)\quad {\rm on } \quad \bbR^3,\label{eq.quasVS1}\\[5pt]
V_s(\Bx,\Go)=\CO(1/|\Bx|)\quad{\rm as}\,\,|\Bx|\to\infty.\label{eq.quasVS2}
\end{numcases}
In the following, we denote respectively by $B(\Go,\delta)$ and $\|\cdot\|_{\infty}$ the open ball of center $\Go$ and radius $\delta$ and the uniform norm on $3\times3$ matrix valued functions defined on the set $\bbR^3$.
 We assume that the dielectric tensor $\BGve(\cdot,\Go)$ satisfies two additional hypotheses in the cloak $\Omega\setminus \CO$:
\begin{itemize}
\item $\tilde{\mathrm{H}}6$ (Uniformly bounded): $\forall \Go \in \operatorname{cl} \bbC^{+}$, $\BGve(\cdot,\Go)$ is a $L^{\infty}$ matrix-valued function on $\Omega\setminus{\CO}$ and it exists a positive constant $c_1$ such that
$
 \sup_{\Go \operatorname{cl}\bbC^{+}}  \|\BGve(\cdot,\Go)\|_{\infty} \leq c_1,
$
\item $\tilde{\mathrm{H}}7$ (Coercivity):
\begin{itemize} 
\item  $\forall \Go \in \bbC^{+}$, it exists $c_2(\Go)>0$ and $\gamma(\Go)\in [0,2 \pi($ such that 
$$
 \  |\Imag(e^{i\, \gamma(\Go)}\BGve(\Bx,\Go)\BE.\overline{\BE})| \geq c_2(\Go)   |\BE|^2, \,  \ \forall \BE \in \bbC^3, \, \mbox{ for a.e. } \Bx \in \bbR^3.
$$
\item $\forall \Go_0 \in \bbR$, $\exists \delta>0$, $c_2(\Go_0)>0$ and $\gamma(\Go_0)\in [0,2 \pi($ such that $\forall \omega \in B(\Go_0,\delta) \cap \operatorname{cl}\bbC^{+}$  :
$$
 \  |\Imag(e^{i\, \gamma(\Go_0)}\BGve(\Bx,\Go)\BE.\overline{\BE})| \geq c_2(\Go_0)   |\BE|^2, \,  \ \forall \BE \in \bbC^3, \, \mbox{ for a.e. } \Bx \in \bbR^3.
$$
Moreover, we suppose that this last property holds also in a neighborhood of $\Go_0=\infty$ by replacing in the previous relation $ B(\Go_0,\delta) $ with $\{ z\in \bbC \mid |z|>1/\delta\}$.\end{itemize}
\end{itemize}
These two hypotheses are classical assumptions. $\tilde{\mathrm{H}}6$ amounts to suppose that the dielectric tensor is uniformly bounded with respect to $\Bx$ and $\Go$, and $\tilde{\mathrm{H}}7$ that it is coercive with respect to $\Bx$. Moreover, we require in $\tilde{\mathrm{H}}7$ the additional property that the constant of coercivity $ c_2(\Go_0)$ holds locally in frequency in a neighborhood of any real frequency or of $\Go=\infty$.  %This is automatically satisfied outside a bounded region since the material outside the cloak is a a standard dielectric but it has to be satisfies inside the cloak.

\begin{Rem}
The coercivity hypothesis $\tilde{\mathrm{H}}7$ is a bit restrictive in the sense that it does not allow any type of passive media, as for instance, a cloak which behaves as a non-dissipative negative index metamaterial whose permittivity $\BGve(\cdot, \Go)$ is a negative constant function. In that particular case, $\BGve(\cdot,\Go)$ changes signs at the boundary of the cloak since the dielectric inclusion and the vacuum have positive permittivity and thus $\tilde{\mathrm{H}}7$ is not satisfied. Nevertheless, this example neglects completely the dissipation of negative index materials which physically allows us to recover $\tilde{\mathrm{H}}7$ even if it is small. However, we think that mathematically, the bounds we derived can be extended to sign-changing media by using mathematical methods associated with sign-changing conductivity equations (see \cite{BonnetBenDhia:2012:TIP,Nguyen:2016:LAP}) which do not require  the coercivity of $\BGve(\cdot,\Go)$.
\end{Rem}

We now look for a solution $V_s(\cdot,\Go)$ of equations (\ref{eq.quasVS1}) and (\ref{eq.quasVS2}) in an appropriate weighted Sobolev space: a Beppo--Levi space, usually used as functional space for solutions of the conductivity equation in unbounded domains. It is defined by 
\begin{equation}\label{eq.defW}
W_{1,-1}(\bbR^3)=\{ u \in S^{'}(\bbR^3) \mid  (1+|\Bx|^2)^{-\frac{1}{2}}\, u \in L^2(\bbR^3) \ \mbox{ and }  \Grad u \in  \BL^2(\bbR^3)\}
\end{equation}
where $S^{'}(\bbR^3)$,  $L^2(\bbR^3)$ and $\BL^2(\bbR^3)$ denote respectively the space of tempered distributions and the spaces of scalar and vector-valued square-integrable functions. 
$ W_{1,-1}(\bbR^3)$ is a Hilbert space (see \cite{Nedelec:2001:AEE}) for the norm
$$
\| u\|_{W_{1,-1}(\bbR^3)}=\| \Grad u \|_{\BL^2(\bbR^3)}=\left(\int_{\bbR^3}|\Grad u|^2 \md \Bx\right)^{\frac{1}{2}}.$$

First, we prove that the equation (\ref{eq.quasVS1}) admits a unique solution in $ W_{1,-1}(\bbR^3)$ which 
depends analytically on $\Go$ on $\bbC^{+}$ and continuously on $\operatorname{cl}{\bbC^{+}}\cup \{\infty\}$. Then we will show that this solution satisfies the asymptotics (\ref{eq.quasVS2}).

We denote by $\Bf(\cdot,\Go)$ the function
$$\Bf(\cdot, \Go)=(\BGve(\cdot,\Go)- \Gve_0 \BI)\BE_0$$ 
which is compactly supported in $\Omega$. Hence, using the hypothesis $\tilde{\mathrm{H}}6$, one checks easily that $\Bf(\cdot, \Go)\in \BL^2(\bbR^3)$.
%the source term of the equation  (\ref{eq.quasVS1}).
By applying the Green formula, it is standard to show that solving the equation (\ref{eq.quasVS1}) in $W_{1,-1}(\bbR^3)$
is equivalent solving the following variational problem:
\begin{equation}\label{eq.formvar}
\mbox{Find } V_s(\cdot,\Go) \in  W_{1,-1}(\bbR^3) \ \mbox{ such that } \ a_{\Go}( V_s(\cdot,\Go),v)=l_{\Go}(v), \  \forall v \in W_{1,-1}(\bbR^3),
\end{equation}
where the sesqulinear form $a_{\Go}$ and the anti-linear form $l_{\Go}$ are respectively defined by
\[
a_{\Go}(u,v)=\int_{\bbR^3} \BGve(\Bx,\Go)\Grad u(\Bx) \cdot \overline{\Grad v(\Bx)} \, \md\Bx  \mbox{ and }  l_{\Go}(v)=\int_{\bbR^3}\Bf(\Bx,\Go) \cdot\overline{ \Grad v(\Bx) }\,  \md \Bx, \forall u, v \in W_{1,-1}(\bbR^3).
\]
With the assumption $\tilde{\mathrm{H}}6$ made on $\BGve(\cdot,\Go)$ and the Cauchy--Schwarz inequality, it is straightforward to show that 
\[
|a_{\Go}(u,v)|\leq  c_1 \, \| u \|_{ W_{1,-1}(\bbR^3)} \, \| v \|_{ W_{1,-1}(\bbR^3)} \ \mbox{  and }\, |l_{\Go}(v)|\leq \|\Bf(\cdot,\Go)\|_{\BL^2(\bbR^3)}  \| v \|_{ W_{1,-1}(\bbR^3)}.
\]
In other words, $a_{\Go}$ and $l_{\Go}$ are continuous. We denote by $ W_{1,-1}(\bbR^3)^{*}$ the dual space of $ W_{1,-1}(\bbR^3)$ and by $\langle \cdot ,  \cdot \rangle  $ the duality product between these two spaces. Classically,  the continuity of $a_{\Go}$ allows us to define a continuous linear operator $\AA(\Go)$ from $W_{1,-1}(\bbR^3)$ to $W_{1,-1}(\bbR^3)^{*}$ by posing:
\[
a_{\Go}(u,v)=\langle \AA(\Go)u, \overline{v} \rangle \mbox{ for all } u, v \in W_{1,-1}(\bbR^3).
\]
The continuity of $l_{\Go}$ proves that $\Bf(\cdot,\Go) \in  W_{1,-1}(\bbR^3)^{*}$ and can be rewritten as $l_{\Go}(v)=\langle \Bf(\cdot,\Go),\overline{ v} \rangle$.
Hence, the variational problem (\ref{eq.formvar}) is equivalent to solving the infinite dimensional system:
\begin{equation}\label{eq.linearsyst}
 \AA(\Go)V_s(\cdot,\Go)=\Bf(\cdot,\Go).
\end{equation}

For two Banach spaces $E$ and $F$, we denote in the following by $\mathcal{L}(E,F)$ the Banach space of bounded linear operators from $E$ to $F$ equipped with the operator norm.
 
\begin{Lem}\label{lem.analyop}
At a fixed frequency $\Go \in \operatorname{cl}\bbC^{+} \cup \{\infty \}$, the operator
$\AA(\Go): W_{1,-1}(\bbR^3)\to W_{1,-1}(\bbR^3)^{*}$ is invertible. 
Moreover, the functions $\Go \mapsto \AA(\Go)$ and $\Go \mapsto \AA(\Go)^{-1}$ defined respectively from $\bbC^{+}$ to $\mathcal{L}\big(W_{1,-1}(\bbR^3), W_{1,-1}(\bbR^3)^{*}\big)$  and from $\bbC^{+}$ to $\mathcal{L}\big((W_{1,-1}(\bbR^3))^{*},W_{1,-1}(\bbR^3)\big)$ 
are analytic for the operator norm.
\end{Lem}
\begin{proof}
Let $\Go \in \operatorname{cl}\bbC^{+} \cup \{\infty \}$.
Thanks to the hypothesis $\tilde{\mathrm{H}}7$ on $\BGve(\Bx,\Go)$, we get that $a_{\Go}$ is a coercive sesquilinear form which satisfies
\[
|a_{\Go}(u,u)| =|e^{i\gamma(\Go)}a_{\Go}(u,u)| \geq |\Imag(e^{i\gamma(\Go)}a_{\Go}(u,u))|\geq  c_2(\Go)  \| u \|_{ W_{1,-1}(\bbR^3)}^2.
\]
Thus, by Lax--Milgram's theorem, $\AA(\Go): W_{1,-1}(\bbR^3) \to  W_{1,-1}(\bbR^3)^{*}$ is an isomorphism. 

One wants now to prove that the function $\Go \mapsto \AA(\Go)$ is  analytic on $\bbC^{+}$ for the operator norm. For this purpose, it is sufficient to show its weak analyticity, in other words that $\Go \mapsto \langle \AA(\Go)u, \overline{v} \rangle$ is analytic on $\bbC^{+}$, for any fixed $u$ and $v\in W_{1,-1}(\bbR^3)$.  As $\BGve(\Bx,\cdot)$ is analytic (by the hypothesis $\tilde{\mathrm{H}}1$), one can check easily by applying the theorem of complex differentiation under the integral presented in \cite{Mattner:2001:CDI} (using the hypothesis $\tilde{\mathrm{H}}6$ for the domination condition required in its assumption) that  $\Go \mapsto \langle \AA(\Go)u, \overline{v} \rangle$ is analytic on $\bbC^{+}$  for any fixed $u$ and $v\in W_{1,-1}(\bbR^3)$. As weak analyticity implies analyticity for the operator norm (see \cite{Kato:1995:PTO}, Theorem 3.12 p. 152), the function $\Go \mapsto \AA(\Go)$ is analytic. Therefore, one deduces (see \cite{Kato:1995:PTO} chapter 7 pp 365-366) that $\Go \mapsto \AA(\Go)^{-1}$ is also analytic for the operator norm. 
\end{proof}

\begin{Thm}\label{analyticity}
At a fixed frequency $\Go \in \operatorname{cl} \bbC^{+}\cup \{\infty \}$, the equation (\ref{eq.quasVS1}) admits
a unique solution $V_s(\cdot,\Go)$ in $ W_{1,-1}(\bbR^3)$ defined by
\begin{equation}\label{eq.sol1}
V_s(\cdot,\Go)=\AA^{-1}(\Go) \Bf(\cdot,\Go).
\end{equation}
Moreover, the function $\Go \mapsto \BE_s(\cdot,\Go)=- \Grad V_s(\cdot,\Go)$ from $\operatorname{cl}\bbC^{+}$ to $\BL^2(\bbR^3)$ equipped with the $\|\cdot\|_{\BL^2(\bbR^3)}$ norm is analytic on $\bbC^{+}$ and continuous on $\operatorname{cl}\bbC^{+}\cup \{\infty \}$. 
 \end{Thm}
 
\begin{proof}
Let $\Go $ be in $\operatorname{cl} \bbC^{+}\cup \{\infty \}$. From  the Lemma \ref{lem.analyop},  we know that the operator $\AA(\Go)$ is invertible. Hence, the equation (\ref{eq.quasVS1}) admits
a unique solution (\ref{eq.sol1})  in $ W_{1,-1}(\bbR^3)$ given by the inversion of the linear system (\ref{eq.linearsyst}). 

Now, we show the analyticity of the function $\Go \mapsto \BE_s(\cdot,\Go)$ on $\bbC^{+}$ for the norm  $\|\cdot\|_{\BL^2(\bbR^3)}$ or equivalently that $\Go \mapsto  \BE_s(\cdot,\Go)=-\nabla V_s(\cdot,\Go)$ is analytic for the norm $\|\cdot\|_{W_{1,-1}(\bbR^3)}$. To achieve this aim, one uses the relation (\ref{eq.sol1}) and  the fact that the function $\Go \mapsto \AA(\Go)^{-1}$  is analytic for the operator norm (see Lemma \ref{lem.analyop}). Thus, it only remains to prove that $\Go \mapsto \Bf(\cdot,\Go)$ is analytic for the norm of $(W_{1,-1}(\bbR^3))^{*}$. By Theorem 1.37 p. 139 of \cite{Kato:1995:PTO}, this is equivalent proving weak analyticity, in other words the analyticity of the functions $\Go \mapsto l_{\Go}(v)= \left \langle  \Bf(\cdot,\Go), v\right \rangle$ for any fixed $v\in W_{1,-1}(\bbR^3)$. This last property is shown once again by applying the theorem of complex differentiation under the integral presented in \cite{Mattner:2001:CDI} (using again the fact that $\BGve(\Bx,\cdot)$ is analytic by the hypothesis $\tilde{\mathrm{H}}1$ and the hypothesis $\tilde{\mathrm{H}}6$ to establish the domination condition required in this theorem).

Thus, it remains to prove the continuity of $\Go \mapsto \BE_s(\cdot,\Go)=-\nabla  V_s(\cdot,\Go)$ for real frequencies and for $\Go=\infty$.  
The reasoning here is slightly different from the one used for the analyticity in the upper half plane. The main reason is that weak continuity does not imply strong continuity.

Let $(\Go_n)$ be a sequence of $\operatorname{cl}\bbC^{+}$ which tends to $\Go \in \bbR$. As $\AA_{\Go_n}\, V_s(\cdot, \Go_n)=\Bf
 (\cdot,\Go_n)$ and $\AA_{\Go}\, V_s(\cdot, \Go)=\Bf(\cdot,\Go)$ (where both  operators $\AA_{\Go_n}$ and $\AA_{\Go}$ are invertible  by Lemma \ref{lem.analyop}), we get the following identity:
$$
V_s(\cdot, \Go)-V_s(\cdot, \Go_n)=\AA_{\Go_n}^{-1} \, \big(\AA_{\Go_n}  V_s(\cdot, \Go)- \AA_{\Go}  V_s(\cdot, \Go)+\Bf(\cdot,\Go)-\Bf(\cdot,\Go_n) \big).
$$
Thus, it follows immediately that:
$$
\|V_s(\cdot, \Go_n)-V_s(\cdot, \Go)\|_{W_{1,-1}(\bbR^3)}\leq \| \AA_{\Go_n} ^{-1}\| \, ( \|(\AA_{\Go_n}  -  \AA_{\Go})  V_s(\cdot, \Go)\|_{W_{1,-1}(\bbR^3)^{*}}+\|\Bf(\cdot,\Go)-\Bf(\cdot,\Go_n) \|_{W_{1,-1}(\bbR^3)^{*}} )
$$
We show now that the right hand side of the last equation tends to zero.
To this aim, one first remarks as a consequence of assumption $\tilde{\mathrm{H}}6$ and Lax-Milgram's Theorem that $ \| \AA_{\Go_n} ^{-1}\| \leq c_2(\Go)^{-1}$ for $n $ large enough. Then, for any $v\in W_{1,-1}(\bbR^3)$, one has, by using the Cauchy--Schwarz inequality:
$$
| \langle (\AA_{\Go}- \AA_{\Go_n} ) V_s(\cdot, \Go), \overline{v} \rangle| \leq \left(\int_{\bbR^3}\| \BGve(\Bx, \Go )- \BGve(\Bx, \Go_n )\|^2 | \nabla V_s(\Bx, \Go)|^2\md \Bx\right)^{\frac{1}{2}}\, \|v\|_{W_{1,-1}(\bbR^3)}.
$$
Thus, using the continuity of $\Gve(\Bx, \cdot )$ at $\Go$ (assumption $\tilde{\mathrm{H}}1$) and the hypothesis $\tilde{\mathrm{H}}6$ for the domination condition, one proves, by applying Lebesgue's dominated convergence theorem, that the integral in the last formula tends to zero. Thus, one concludes that:
$$
 \|(\AA_{\Go_n}  -  \AA_{\Go})  V_s(\cdot, \Go)\|_{W_{1,-1}(\bbR^3)^{*}} =\sup_{v\in W_{1,-1}(\bbR^3)\setminus\{0\}} \frac{\big | \big \langle (\AA_{\Go}- \AA_{\Go_n} ) V_s(\cdot, \Go), \overline{v}  \big\rangle \big|}{\|v\|_{W_{1,-1}(\bbR^3)}}\to 0, \mbox{ as } n \to +\infty.
$$
Finally, by doing the same reasoning for the term $\|\Bf(\cdot,\Go)-\Bf(\cdot,\Go_n) \|_{W_{1,-1}(\bbR^3)^{*}}$, one has:
$$
|\langle \Bf(\cdot,\Go)-\Bf(\cdot,\Go_n) ,v \rangle| \leq \left( \int_{\Omega}\| \BGve(\Bx, \Go )- \BGve(\Bx, \Go_n )\|^2 | |E_0|^2 \md \Bx \right)^{\frac{1}{2}} \|v\|_{W_{1,-1}(\bbR^3)},
$$
Then, by using once again Lebesgue's dominated convergence theorem (thanks to the assumptions $\tilde{\mathrm{H}}1$ and $\tilde{\mathrm{H}}6$), one shows also that:
$
\|\Bf(\cdot,\Go)-\Bf(\cdot,\Go_n)\|_{W_{1,-1}(\bbR^3)^{*}} \to 0, \mbox{ as } n \to +\infty.
$
Thus, one concludes that:
 $$
 \|V_s(\cdot, \Go_n)-V_s(\cdot, \Go)\|_{W_{1,-1}(\bbR^3)}\to 0, \mbox{ as } n \to +\infty.
  $$  
The same proof as above holds to show the continuity at $\Go=\infty$.  
\end{proof}

Now, we state why the solution  $V_s(\cdot,\Go)$  of (\ref{eq.sol1}) satisfies not only the equation (\ref{eq.quasVS1})  
but also admits the asymptotic expansion (\ref{eq.quasVS2}) and more precisely that the leading term of $V_s(\cdot,\Go)$  at infinity is a dipolar field (see \cite{Jackson:1999:CE}).
To this aim, one uses the fact that outside the cloaking device, the equation (\ref{eq.quasVS1}) becomes the following Laplace equation:
\begin{equation}\label {eq.Dirichletext}
\nabla^2 u(\cdot,\Go)=0 \ \mbox{ on } \bbR^3\setminus B(0,R)   \mbox{ and } u(\cdot,\Go)=V_s(\cdot,\Go)  \mbox{ on  }\partial B(0,R)
\end{equation}
for $R>R_0$, where $\partial B(0,R)$ denotes the boundary of a ball $B(0,R)$ which contains and does not intersect the cloaking device $\Omega$. We point out that  the trace  of $V_s(\cdot,\Go)$ on the sphere $\partial B(0,R)$ belongs to $H^{1/2}(\partial B(0,R))$ since $V_s(\cdot,\Go)$ is locally $H^{1}$ (indeed, by using standard interior regularity results for second order elliptic equations, see for instance Theorem 2 p. 314 of \cite{Evans:2008:PDE}, one can show that the trace of $V_s(\cdot,\Go)$ belongs to any $H^{s}\Big(\partial B(0,R)\big)$ for $s>0$). Therefore, the Dirichlet exterior problem (\ref{eq.Dirichletext})  (see for example Theorem 2.5.14 of \cite{Nedelec:2001:AEE}) admits a unique solution in $W_{1,-1}(\bbR^3\setminus B(0,R) )$  given by the restriction of $V_s(\cdot,\Go)$ on $\bbR^3\setminus B(0,R)$ (the definition of the space $W_{1,-1}(\bbR^3\setminus B(0,R) )$ is deduced from the definition (\ref{eq.defW})  by replacing $\bbR^3$ by $\bbR^3\setminus B(0,R)$). Moreover, as a solution of the Laplace equation (\ref{eq.Dirichletext}) this solution admits the following integral representation:
$$
V_s(\Bx,\Go)=\int_{\partial B(0,R)}   \frac{\partial G(\Bx,\By) } {\partial \Bn_y}  V_s(\By,\Go)- G(\Bx,\By) \frac{ \partial V_s(\By,\Go) }{\partial \Bn_y} \, \md \By  \, \mbox{ with }  \, G(\Bx,\By)=\frac{1}{4\pi \, |\Bx-\By|},
$$
where $\Bn_y=\By/|\By|$ is the outward normal of the domain $B(0,R)$.
Then, by using the asymptotic expansions of the Green function $ G(\Bx,\By)$ and its normal derivative $\partial G(\Bx,\By)/ \partial \Bn_y$ for large values of $|\Bx|$:
$$
G(\Bx,\By)=\frac{1}{4\pi |\Bx|}+\frac{\By\cdot \Bx}{4\pi |\Bx|^3}+O\left(\frac{1}{|\Bx|^3}\right) \ \mbox{ and } \ \frac{\partial G(\Bx,\By) } {\partial \Bn_y}=\frac { \Bn_y\cdot \Bx}{4\pi |\Bx|^3}+O\left(\frac{1}{|\Bx|^3}\right)
$$
which holds uniformly in $\By$ on $\partial B(0,R)$, this leads to
\begin{equation}\label{eq.asympVS}
V_s(\Bx,\Go)=\frac{Q(\Go)}{4\pi \Gve_0 |\Bx| }+\frac{\Bp(\Go) \cdot \Bx }{4\pi  \Gve_0 |\Bx|^3 }+\CO\Big(\frac{1}{|\Bx|^3}\Big),
\end{equation}
where the charge (also called monopole term) $Q(\Go)$ and the induced dipole moment $\Bp(\Go)$ are respectively given by
\begin{equation}\label{eq.monopdipol}
 Q(\Go)=   \Gve_0 \int_{\partial B(0,R)}  \frac{- \partial V_s(\By,\Go) }{\partial \Bn_y} \md \By \ \mbox{ and }\ \Bp(\Go)= \Gve_0 \int_{\partial B(0,R)}\frac{ -\partial V_s(\By,\Go) }{\partial \Bn_y} \, \By+ V_s(\By,\Go)  \Bn_y \, \md \By.
\end{equation}
In our scattering problem, one can easily prove that the monopole term $Q(\Go)$ vanishes.
Indeed, using the divergence theorem and the equation (\ref{eq.quasVS1}), one gets:
$$
\int_{\partial B(0,R)} \Gve(\By,\Go) \, \frac{ \partial V(\By,\Go) }{\partial \Bn_y} \md \By=\int_{\partial B(0,R)} \Gve_0 \, \frac{ \partial V(\By,\Go) }{\partial \Bn_y} \md \By=0 
$$
and thus by virtue of (\ref{eq.pot}) and (\ref{eq.monopdipol}), one obtains that
$$
 Q(\Go)= - \int_{\partial B(0,R)} \Gve_0 \BE_0 \cdot \Bn \, \md \By=-\Gve_0 \BE_0 \cdot \int_{\partial B(0,R)}  \Bn \, \md \By=0 \,.
$$
Hence, the leading term of the scattered field (\ref{eq.asympVS}) is a dipolar term and using (\ref{eq.quasVS1}), one finally gets a justification of the asymptotic formula (\ref{eq.asymp}) for the potential $V$. Our aim is now to derive a more explicit expression for the induced dipole moment $\Bp(\Go)$ by a small computation done in \cite{Milton:2002:TC} that we reproduce here for readability. Let  $\BE_0$ be any vector of $\bbC^3$. Then, from relations (\ref{eq.pot}), (\ref{eq.quasVS1}), (\ref{eq.monopdipol}) and the Green identity, one has:
\begin{eqnarray*}
\Bp(\Go)\cdot \overline{\BE_0}&= & \int_{\partial B(0,R)} -\Gve_0 \frac{ \partial V_s(\By,\Go) }{\partial \Bn_y} \, \By\cdot \overline{\BE_0}+ \Gve_0 V_s(\By,\Go)  \Bn_y \cdot \overline{\BE_0}  \,\md \By\\
&=& \int_{\partial B(0,R)}\Big(-\Gve_0 \frac{ \partial V(\By,\Go) }{\partial \Bn_y} \Big)\, (\By\cdot \overline{\BE_0}) + \Gve_0 V(\By,\Go)  \Bn_y \cdot \overline{\BE_0}  \,\md \By\\
&=& \int_{B(0,R)} \BGve(\By,\Go) \BE(\By,\Go) \cdot \nabla (\By\cdot \overline{\BE_0}) \,\md \By- \Big(\int_{B(0,R)} \Gve_0 \BE(\By,\Go)   \,\md \By \Big) \cdot \overline{\BE_0}  \\
&=&\int_{B(0,R)} \big( \BGve(\By,\Go)- \Gve_0 \BI\big)  \BE(\By,\Go) \md \By  \cdot \overline{\BE_0}  .
\end{eqnarray*}

As the function $\BGve(\cdot,\Go)- \Gve_0 \, \BI$ has a support contained in the cloaking device, this yields:
\begin{equation}\label{eq.polaztens}
\Bp(\Go)=\BGa(\Go) \BE_0=\int_{\Omega} \big( \BGve(\By,\Go)- \Gve_0 \BI\big)  \BE(\By,\Go) \md \By, \, \forall \Go\in \operatorname{cl}\bbC^{+}\cup \{\infty\} \, .
\end{equation}
Now, using this last  formula, one can rewrite the function $f$ defined by (\ref{eq.falpha}) as
\begin{equation}\label{eq.defpola2}
f(\Go)=\BGa(\Go) \BE_0. \overline{\BE_0}= \int_{\Omega} (\BGve(\Bx,\Go)-\BGve_0\BI ) \, \BE(\Bx, \Go) \cdot  \overline{\BE_0} \, \md \Bx, \, \forall \Go \in  \operatorname{cl}\bbC^{+}\cup \{\infty\}
\end{equation}
and study its regularity with respect to the frequency $\Go$.
\begin{Pro}\label{prop.analyf}
For any fixed incident field $\BE_0\in \bbC^3$, the function $f$ defined by (\ref{eq.falpha}) is analytic on $\bbC^{+}$ and continuous on $\operatorname{cl}\bbC^{+}$, in other words, it satisfies the hypothesis $\mathrm{H}1$. 
\end{Pro}

\begin{proof}
Let $\Go\in \bbC^{+}$ and  $\BE_0$ be a fixed vector of $\bbC^3$ . One introduces the linear form: 
$$
L_{\Go}(\BU)=\int_{\Omega} (\BGve(\Bx,\Go)-\BGve_0\BI ) \BU \cdot  \overline{\BE_0} \, \md \Bx, \ \forall  \BU \in L^{2}(\Omega)
$$
such that $f(\Go)=L_{\Go}\big(\BE(\cdot,\Go)\big)$.
One easily checks that $L_{\Go}$ is well-defined and continuous by virtue of $\tilde{\mathrm{H}}6$. Moreover, using again the theorem of complex differentiation under the integral presented in \cite{Mattner:2001:CDI} (and both the hypotheses  $\tilde{\mathrm{H}}1$ and $\tilde{\mathrm{H}}6$ to prove respectively the regularity and the domination condition required in the assumptions of this theorem), one shows that $L_{\Go}$ is weakly analytic, in other words, for any fixed $\BU\in  L^{2}(\Omega)$, $\Go \to L_{\Go}(\BU)$ is analytic on $\bbC^{+}$. Hence, as weak analyticity implies strong analyticity (see \cite{Kato:1995:PTO}, Theorem 1.37 p.139), $\Go\to L_{\Go}$ is also analytic. Then, using Theorem \ref{analyticity}, one has $\Go \to \BE(\cdot,\Go)=\BE_0+ \BE_s(\cdot,\Go)$ is strongly analytic on $\bbC^+$ for the $L^2(\Omega)$ norm and one finally deduces that $\Go \to f(\Go)=L_{\Go}\big(\BE(\cdot,\Go)\big)$ is analytic on $\bbC^+$.

It remains to prove the continuity of $f$ for a real frequency $\Go$. Let $(\Go_n)$ be a sequence of  in $\operatorname{cl}\bbC^{+}$ which tends to $\Go \in \bbR$.
One has:
$$
|f(\Go_n)-f(\Go)  |\leq  |L_{\Go_n}\big(\BE(\cdot,\Go_n)-\BE(\cdot,\Go)\big) |+ |L_{\Go_n}\big(\BE(\cdot,\Go)\big)-L_{\Go}\big(\BE(\cdot,\Go)\big)|.
$$
By $\tilde{\mathrm{H}}6$, one checks easily that the linear form $L_{\Go}$ is uniformly bounded with respect to the frequency: $\|L_{\Go_n}\|\leq (c_1+\Gve_0| )\, |\BE_0| \operatorname{mes}(\Omega)^{1/2}$ and by Theorem \ref{analyticity}, $\Go \to \BE(\cdot,\Go) $ is continuous on $\operatorname{cl}\bbC^{+}$ for the $L^2(\Omega)$ norm, thus the first term of the right hand side tends to $0$. Concerning the second term, by using hypotheses $\tilde{\mathrm{H}}1$ and $\tilde{\mathrm{H}}6$, it is straightforward to check , by applying Lebesgue's dominated convergence theorem, that it tends to $0$. This concludes the proof. 
\end{proof}

\begin{Rem}
From the analyticity of the function $f$ on $\bbC^{+}$ for any fixed incident field $\BE_0\in \bbC^3$, one deduces since weak analyticity implies
analyticity for the operator norm (see \cite{Kato:1995:PTO}, Theorem 3.12 p. 152) that the polarizability tensor $\BGa$ is an analytic function of the frequency on $\bbC^+$ with respect to the induced $l^2-$norm.
\end{Rem}

We prove in the following proposition that $f$ satisfies the hypotheses H2 and H3.

 \begin{Pro}\label{prop.hyp2-3}
For any fixed non-zero incident field $\BE_0\in \bbC^3$, the function $f$ defined by (\ref{eq.falpha}) satisfies 
$$
f(\Go) \to f(\infty)=\Ga(\infty)\BE_0 \cdot \overline{\BE_0}>0,
$$
where $\Ga(\infty)$ is defined by relation (\ref{eq.defpola2}) evaluated at $\Go=\infty$. Thus, $f$ satisfies  $\mathrm{H}2$. Moreover, $f$ satisfies the hypothesis $\mathrm{H}3$, that is
\begin{equation}\label{eq.fsym}
f(-\overline{\Go})=\overline{f(\Go)},  \, \forall \Go \in \operatorname{cl}\bbC^{+}.
\end{equation}
When the cloak contains a non reciprocal medium (in other words if $\tilde{\mathrm{H}}5$ is not satisfied), the relation (\ref{eq.fsym}) holds only under the additional assumption that $\BE_0$ is a real-valued incident field.
 \end{Pro}
 \begin{proof}
Let $\BE_0$ be a fixed vector of $\bbC^3$. The fact that $f(\Go) \to f(\infty)$ amounts to proving the continuity of $f$ at $\Go=\infty$, which can be dealt  with in the same way as the continuity of $f$ for a real frequency $\Go$ in the proof of Proposition \ref{prop.analyf}. The positivity of the limit $f(\infty)$ which has to be proved for any $\BE_0 \in \bbC^3$ turns out to be equivalent to showing that  $\Ga(\infty)$ is positive-definite. This last result is a well-known property (see Theorem 4.11 of \cite{Ammari:2007:PMT}) of the polarizability tensor associated with a homogeneous simple connected isotropic inclusion $\CO$ with Lipschitz boundary embedded in the vacuum and defined by its real permittivity $\Gve>\Gve_0$.

One wants now to show relation (\ref{eq.fsym}). Using $\tilde{\mathrm{H}}3$, one deduces by uniqueness of the solution of (\ref{eq.quasVS1}) and
(\ref{eq.quasVS2}) in $W_{1,-1}(\bbR^3)$, that $V_s(\cdot,-\overline{\Go})$, the solution of these equations corresponding to an incident field $\overline{\BE_0}$, is equal to $\overline{V_s(\cdot,\Go)}$ (where $V_s(\cdot,\Go)$ stands for the solution with incident field $\BE_0$). Thus, it follows using $\tilde{\mathrm{H}}3$ and the relations (\ref{eq.pot}) and (\ref{eq.polaztens}) that $\BGa(-\overline{\Go})=\overline{\BGa(\Go)}$. Finally,  one deduces from this last equality that
$$
f(-\overline{\Go})=\overline{\BGa(\Go)} \BE_0 \cdot \overline{\BE_0} \, .
$$
In the particular case of  a real-valued field $\BE_0$, this implies directly the relation $(\ref{eq.fsym})$. 
Obtaining the same relation for non real-valued incident fields requires the symmetry of the polarizability tensor: $\BGa(\Go)^{\top}=\BGa(\Go)$ which is a consequence of the reciprocity principle $\tilde{\mathrm{H}}5$ (for its proof see for instance \cite{Sohl:2008:DRS} p. 62). This concludes the proof.
 \end{proof}
 
 Finally, we have to show that $f$ satisfies the hypothesis H4. This is the purpose of the following proposition.
     
 \begin{Pro}\label{prop.hyp4}
For any fixed incident field $\BE_0\in \bbC^3$, the function $f$ defined by (\ref{eq.falpha}) satisfies the hypothesis $\mathrm{H}4$:
$$
\operatorname{Im}f(\Go)\geq 0, \ \forall \Go\in \bbR^{+}.
$$
Moreover, if $[\Go_-,\Go_+]$ is a transparency window, that is a frequency band for which  $\operatorname{Im} \BGve(\Bx,\cdot)=0$  for a. e. $\Bx \in \Omega\setminus \CO$ then $$\operatorname{Im}f(\Go)= 0, \ \forall \Go \in [\Go_-,\Go_+].$$
 \end{Pro}
 
 \begin{proof}
 Let $\BE_0$ be a fixed vector of  $\bbC^3$, $\Go$ a non-negative frequency and $R$ a positive real number satisfying $R>R_0$. We denote by $\BE$ the electrical field associated with the incident field $\BE_0$ by the equation (\ref{eq.quasis}). Then, by using $\tilde{\mathrm{H}}4$, one has that:
 $$
I(R,\Go)=\int_{B(0,R)} \operatorname {Im} \BGe(\Bx,\Go) \BE(\Bx,\Go)\cdot\overline{ \BE(\Bx,\Go)} \, \md \Bx=\operatorname {Im}\Big( \int_{B(0,R)} \BGe(\Bx,\Go) \BE(\Bx,\Go)\cdot\overline{ \BE(\Bx,\Go) }\, \md \Bx\Big) \geq 0.
 $$ 
We point out that $I(R,\Go)$ is constant with respect to $R$ for $R>R_0$ since $ \operatorname {Im} \BGe(\cdot,\Go)$ is compactly supported in $\Omega \subset B(0,R)$. Using the Green identity and the facts that $\BE=-\nabla \BV$ is divergence free (see \ref{eq.quasis}) and that $ \BGve(\Bx,\Go)=\Gve_0 \, \BI$ on $\partial B(0,R)$, one gets:
$$
I(R,\Go)=  \operatorname {Im} \Big(\int_{\partial B(0,R)}    \BV(\Bx,\Go) \cdot \Gve_0 \, \overline{  \BE(\Bx,\Go)\cdot \Bn}\, \md \Bx \Big) \geq 0.
 $$
 From the asymptotics (\ref{eq.asymp}) of $ \BV(\Bx,\Go)$ and the asymptotics of the electric field
 $$
  \BE(\Bx,\Go)=-\nabla V (\Bx,\Go) =\BE_0-\frac{\Bp(\Go)}{4\pi \Gve_0 |\Bx|^3}+3 \frac{(\Bp(\Go)\cdot \Bx) \,  \Bx }{4 \, \pi \Gve_0 |\Bx|^5}+O\left(\frac{1}{|\Bx|^4}\right)
 $$
one gets:
 $$
 I(R,\Go)=3 \,\operatorname {Im} \Big(\int_{\partial B(0,R)} \frac{(\Bp(\Go)\cdot \Bx)}{4 \, \pi \Gve_0 |\Bx|^4} \cdot \overline{\BE_0\cdot \Bx} \, \md \Bx\Big)+O\left(\frac{1}{|\Bx|}\right) \geq 0.
 $$
Using the fact that $\Bp(\Go)=\BGa(\Go) \BE_0$ leads to 
$$
\operatorname{Im} \Big(\int_{\partial B(0,R)} \frac{( \BGa(\Go) \BE_0 \cdot \Bx)  }{4 \, \pi \Gve_0 |\Bx|^4} \cdot \overline{\BE_0\cdot \Bx} \, \md \Bx\Big)\geq 0 \, .
$$
By virtue of the algebraic identity $$( \BGa(\Go) \BE_0 \cdot \Bx)\cdot \overline{\BE_0\cdot \Bx} = \Bx \Bx^{\top} \BGa(\Go)\BE_0 \cdot \overline{ \BE_0} \, ,$$ one finally obtains
$$
\operatorname{Im} \Big(\int_{\partial B(0,R)} \frac{ \Bx \Bx^{\top} \BGa(\Go) \BE_0 \cdot \overline{ \BE_0}}{4 \, \pi \Gve_0 |\Bx|^4}   \md \Bx\Big)= \frac{1}{3\Gve_0} \operatorname{Im} (\BGa(\Go)\BE_0 \cdot \overline{\BE_0 }) =\frac{1}{3\Gve_0} \operatorname{Im}f(\Go)\geq 0 \, .
$$
(to derive the last inequality, we use the identity: $\int_{\partial B(0,R)} \Bx \Bx^{\top} \md \Bx=4\, \pi/3 \, \BI$ which can be shown by a straightforward computation).

 In the particular case where $\operatorname{Im} \BGve(\Bx,\Go)=0$ for a. e. $\Bx\in \Omega \setminus \CO$, one has $I(R,\Go)=0$ and thus one deduces that $\operatorname{Im}f(\Go)=0$. This concludes the proof.

 \end{proof}
One concludes this subsection by the following theorem which defines a Herglotz function associated with the polarizability tensor.

\begin{Thm}\label{th.herglcloak}
Let $\BE_0$ be a non-zero fixed vector of $\bbC^3$ and $f$ the function defined by (\ref{eq.falpha}). 
If the cloak $\Omega\setminus \CO$ satisfies the hypotheses $\tilde{\mathrm{H}}1-7$, then
\begin{equation}\label{eq.vcloaking}
v(\Go)=\Go f(\sqrt{\Go})=\Go \BGa(\sqrt{\Go}) \BE_0\cdot \overline{ \BE_0}
\end{equation}
is a Herglotz function which is analytic on $\bbC\setminus{\bbR^{+}}$ and negative on $\bbR^{-*}$. Moreover, in its representation given by the Theorem \ref{thm.Herg}, the measure $\mm$ is supported in $\bbR^{+}$ and $\alpha$ is equal to $f_{\infty}=\BGa(\infty) \BE_0 \cdot \overline{\BE_0}$. 
\end{Thm}

\begin{proof}
The proof is an immediate consequence of propositions  \ref{prop.analyf}, \ref{prop.hyp2-3} and \ref{prop.hyp4} which show that the function $f$ satisfies the hypotheses H1-4 and thus the function $v$ defined by (\ref{eq.vcloaking}) satisfies the Corollary \ref{cor.herg}. 
\end{proof}

\begin{Rem}
As the reciprocity principle $\tilde{\mathrm{H}}5$ is only required to prove $\mathrm{H}3$ for non real incident fields $\BE_0$ (see Proposition \ref{prop.hyp2-3}), Theorem \ref{th.herglcloak} still holds for non-reciprocal media under the additional condition that $\BE_0\in \bbR^3$. 
\end{Rem}

\subsection{Fundamental limits of broadband passive cloaking in quasi-statics}\label{sec.boundcloaking}

\subsubsection{General bounds on the polarizability tensors}

One assumes in the following that our cloak satisfies the hypotheses $\tilde{\mathrm{H}}1-7$.
One wants first to establish that the analyticity property of the polarizability tensor is sufficient to prove that $\BGa$ does not vanish on the whole frequency band $[\Go_-,\Go_+]$. 
Indeed, for any fixed non zero incident field $\BE_0 \in \bbC^3$, the function $f$ defined  by (\ref{eq.falpha}) is analytic on $\bbC^{+}$ and continuous on $\operatorname{cl}\bbC^{+}$. Thus, if by contradiction $\BGa$ vanishes on $[\Go_-,\Go_+]$, so does $f$. Then, using the Schwarz reflection principle and the analytic continuation, one deduces that $f$ vanishes also on the whole upper-half plane $\bbC^{+}$, which contradicts the fact that $f$ tends to $f_{\infty}>0$, when $|\Go|\to \infty$.
However, such analytic continuation arguments are not of practical interest. Indeed, it is possible for instance for a polynomial to be arbitrary close to $0$ on one disk and arbitrarily closed to $1$ on another disjoint disk (see \cite{Vasquez:2012:MAT}). But, knowing that one can construct a Herglotz function associated with $f$ (see Theorem \ref{th.herglcloak}) gives us meaningful inequalities by using the bounds derived in section \ref{sec-boundanalycomp}. These inequalities establish fundamental limits on the cloaking effect over a frequency band $[\Go_-,\Go_+]$, that we present in this subsection.

For any non-zero incident field $\BE_0\in \bbC^3$, the function $v$ defined by (\ref{eq.vcloaking}) satisfies the bound (\ref{eq.sumrulesmu}). In the particular case of Dirac measures: $m=\delta_{\xi}, \, \xi \in \bbR$, which optimizes the right hand side of (\ref{eq.sumrulesmu}) (see Theorem \ref{thm.mesopt}), this inequality becomes the bound (\ref{eq.sumruledelta}), that we are recalling here in the case of our cloaking application:
\begin{equation}\label{eq.transparencywindowcloak}
\lim_{y\to 0^{+}}  \int_{x_-}^{x_+} \Imag \left( \displaystyle \frac{1}{\xi-(x +\ii y)\BGa(\sqrt{x+\ii y})\BE_0\cdot \overline{\BE_0}}\right) \, \md x \leq  \frac{\pi}{\BGa(\infty)\BE_0\cdot \overline{\BE_0}}, \  \forall \xi \in \bbR,
\end{equation}
which holds for any interval $[x_-,x_+]$ of  $\bbR^{+*}$. One notices that the geometry and the dielectric contrast of the inclusion $\CO$ are encoded in the expression of the polarizability tensor $\BGa(\infty)$. For instance, in the case of circular inclusion of radius $R$, one has
$$\BGa(\infty)=4\pi R^3  \Gve_0 \, \frac{\Gve-\Gve_0}{\Gve+2\Gve_0} \BI\, .$$
Explicit expressions of $\BGa(\infty)$ can be also derived for ellipsoidal shapes (see for instance \cite{Osborn:1945:DFG,Stoner:1945:DFE}).

\subsubsection{The case of a transparency window}
We are now interested in deriving a more explicit version of the bound (\ref{eq.transparencywindowcloak})
in the case  where $[\Go_-,\Go_+]$ is a transparency window, that is a frequency band for which $\operatorname{Im} \BGve(\Bx,\cdot)=0$  for a. e. $\Bx$ in $\Omega\setminus \CO$. In other words, we assume that the cloak material is composed of a material that one can consider lossless in this frequency range. In particular, this latter condition implies that the polarizability tensor satisfies also $\operatorname{Im} \BGa(\Go)=0$ on $[\Go_-,\Go_+]$ (see Proposition \ref{prop.hyp4}). Thus, one can directly apply the bound (\ref{eq.boundtranspf}) derived  in Proposition \ref{Pro.boundtransp}  to obtain
\begin{equation}\label{eq.boundscalarlevel}
\Go_0^2 \, (\BGa(\Go_0)-\BGa(\infty))\BE_0 \cdot \overline{\BE_0} \leq \Go^2 (\BGa(\Go)-\BGa(\infty))\BE_0 \cdot \overline{  \BE_0}, \  \forall \Go, \Go_0 \in [\Go_-,\Go_+] \mbox{ such that } \ \Go_0\leq \Go.
\end{equation}
As, this last inequality holds for any $\BE_0\in \bbC^3$, one deduces that
\begin{equation}\label{eq.transpapolabound}
\Go_0^2 \, (\BGa(\Go_0)-\BGa(\infty)) \leq \Go^2 \,(\BGa(\Go)-\BGa(\infty)), \, \  \forall \Go, \Go_0 \in [\Go_-,\Go_+] \mbox{ such that } \ \Go_0\leq \Go
\end{equation}
which is to be interpreted as a matrix inequality ($\BA\leq\BB$ if and only if $\BB-\BA$ is positive semidefinite).
We want to emphasize that this bound is sharp in the sense that there exists an analytic function $\BGa$ (given by a Drude type model) such that $f(\cdot)=\BGa(\cdot)\BE_0\cdot \overline{\BE_0}$ satisfies the properties H1-4  (except the continuity at $\Go=0$) for any $\BE_0\in \bbC^3$, namely
\begin{equation*}
\BGa(\Go)=\BGa(\infty)-\frac{\Go_0^2[\BGa(\infty)-\BGa(\Go_0)]}{\Go^2} \ \mbox{ with } \ \BGa(\Go_0)\leq \BGa(\infty)
\end{equation*}
for which one has equality in (\ref{eq.transpapolabound}). This function is singular at $\Go=0$, so it does not satisfy completely the hypothesis H1, but the continuity assumption on the real line in H1 can be weakened as we point out in Remarks \ref{Rem.DrudeLorentzmodel} and \ref{Rem.weakhypcont}.

Now coming back to our initial cloaking problem, if one can cloak the inclusion at one frequency $\Go_0\in [\Go_-,\Go_+]$, and thus if  $\BGa(\Go_0)=0$ then the bound (\ref{eq.transpapolabound}) implies
\begin{eqnarray*}
 \BGa(\Go) & \leq &  -\BGa(\infty) \frac{\Go^2_0-\Go^2}{\Go^2} \  \mbox{ if }  \ \Go_-\leq  \Go\leq \Go_0\\[5pt]
 & \geq & \BGa(\infty) \frac{\Go^2-\Go^2_0}{\Go^2}\ \mbox{ if } \  \Go_0 \leq \Go \leq \Go_+ \, ,
\end{eqnarray*}
which obviously forces $\BGa(\Go)$ to be non-zero away from the frequency $\Go_0$ (provided one is still in the transparency window $[\Go_-,\Go_+]$ where there is no absorption). Thus, one cannot achieve broad  band passive cloaking in a transparency window.

\subsubsection{The lossy case}
The bound (\ref{eq.transpapolabound}) is only valid if $[\Go^-,\Go^+]$ is a transparency window and thus does not hold if the cloak is a lossy material over this frequency range. Nevertheless, for a lossy cloak, one can apply the bounds (\ref{eq.sumrulelog}), (\ref{eq.bound1losscase}), (\ref{eq.bound2losscase}) derived in section \ref{lossy-mat} to the function $f$. In particular, the bound (\ref{eq.bound2losscase}) takes the form:
\begin{equation}\label{eq.bound4losscase}
\frac{1}{4}(\Go_+^2-\Go_-^2) \BGa(\infty)\BE_0 \cdot \overline{\BE_0} \leq \max_{x\in[\Go_-,\Go_+]}|\omega^2 \BGa(\Go)\BE_0 \cdot \overline{\BE_0}|, \, \forall \BE_0 \in \bbC^3.
\end{equation}
This bound gives a limitation to the cloaking effect by controlling from below the maximum of the function $ \Go \to \omega^2 \BGa(\Go)\BE_0 \cdot \overline{\BE_0}$ by a positive quantity depending both on the frequency bandwidth: $\Go_+-\Go_-$ and on the geometry and the dielectric contrast of the inclusion with the term $\BGa(\infty)\BE_0 \cdot \overline{\BE_0}$.
However, from an experimental perspective the more general bounds (\ref{eq.sumrulelog}) or (\ref{eq.bound1losscase}) are more
meaningful since the value of the left hand side of (\ref{eq.bound4losscase}) would be drastically changed if there was an extremely narrow resonant spike in $f$ in the considering interval and such a spike would 
be difficult to experimentally detect. 

\begin{Rem}
For the sake of generality, 
we point out that all the bounds derived in this subsection, with the exception of (\ref{eq.transpapolabound}), which does not hold at the tensor level (but only at the scalar level  (\ref{eq.boundscalarlevel})) are still satisfied for real-valued incident fields $\BE_0$ if the medium does not satisfy the reciprocity principle $\tilde{\mathrm H}5$. Thus, one has proved also that one cannot achieve broadband cloaking with non reciprocal materials. 
\end{Rem}

\section*{Acknowledgements}
G.W. Milton is grateful to the Mittag-Leffler Institute for hosting his visit to Sweden during the program on Inverse Problems and Applications, where this work was initiated, and both authors are grateful to the
National Science Foundation for support through grant DMS-1211359. Additionally, they are grateful to the Institute for Mathematics and its applications in Minneapolis for supporting their visit there in the Fall 2016.

\bibliographystyle{siam}
\bibliography{tcbook,newref}

\ifx \bblindex \undefined \def \bblindex #1{} \fi\ifx \bblindex \undefined \def
  \bblindex #1{} \fi\ifx \bbljournal \undefined \def \bbljournal #1{{\em
  #1}\index{#1@{\em #1}}} \fi\ifx \bblnumber \undefined \def \bblnumber #1{{\bf
  #1}} \fi\ifx \bblvolume \undefined \def \bblvolume #1{{\bf #1}} \fi\ifx
  \noopsort \undefined \def \noopsort #1{} \fi
\begin{thebibliography}{10}

\bibitem{Alu:2005:ATP}
{\sc A.~Al{\'u} and N.~Engheta}, {\em Achieving transparency with plasmonic and
  metamaterial coatings}, Physical Review E (Statistical physics, plasmas,
  fluids, and related interdisciplinary topics), 72 (2005), p.~0166623.

\bibitem{Alu:2008:PMC}
\leavevmode\vrule height 2pt depth -1.6pt width 23pt, {\em Plasmonic and
  metamaterial cloaking: physical mechanisms and potentials}, Journal of Optics
  A: Pure and Applied Optics, 10 (2008), p.~093002.

\bibitem{Ammari:2013:ALR}
{\sc H.~Ammari, G.~Ciraolo, H.~Kang, H.~Lee, and G.~W. Milton}, {\em Anomalous
  localized resonance using a folded geometry in three dimensions}, Proceedings
  of the Royal Society A: Mathematical, Physical, \& Engineering Sciences, 469
  (2013), p.~20130048.
\newblock Also available as arXiv:1301.5712 [math-ph].

\bibitem{Ammari:2013:STN}
\leavevmode\vrule height 2pt depth -1.6pt width 23pt, {\em Spectral theory of a
  {Neumann--Poincar{\'e}}-type operator and analysis of cloaking due to
  anomalous localized resonance}, Archive for Rational Mechanics and Analysis,
  208 (2013), pp.~667--692.
\newblock See also arXiv:1109.0479 [math.AP].

\bibitem{Ammari:2013:STNII}
\leavevmode\vrule height 2pt depth -1.6pt width 23pt, {\em Spectral theory of a
  {Neumann--Poincar{\'e}}-type operator and analysis of cloaking due to
  anomalous localized resonance {II}}, Contemporary Mathematics, 615 (2014),
  pp.~1--14.

\bibitem{Ammari:2007:PMT}
{\sc H.~Ammari and H.~Kang}, {\em Polarization and moment tensors: with
  applications to inverse problems and effective medium theory}, vol.~162,
  Springer Science \& Business Media, New York, 2007.

\bibitem{Ando:2015:SPN}
{\sc K.~Ando, Y.-G. Ji, H.~Kang, K.~Kim, and S.~Yu}, {\em Spectral properties
  of the neumann-poincar{\'e} operator and cloaking by anomalous localized
  resonance for the elasto-static system},  (2015).
\newblock Submitted. Available as arXiv:1510.00989 [math.AP].

\bibitem{Baker:1981:PAB}
{\sc G.~A. {Baker Jr.} and P.~R. Graves-Morris}, {\em {Pad{\'e}} Approximants:
  Basic Theory. {Part I}. Extensions and Applications. {Part II}}, vol.~13 \&
  14 of Encyclopedia of Mathematics and its Applications,
  Ad{\-d}i{\-s}on-Wes{\-l}ey, Reading, Massachusetts, 1981.

\bibitem{Berg:2008:SPBS}
{\sc C.~Berg}, {\em {Stieltjes--Pick--Bernstein--Schoenberg} and their
  connection to complete monotonicity}, in Positive Definite Functions: From
  {Schoenberg} to Space--Time Challenges, J.~Mateu and E.~Porcu, eds.,
  Editorial Universitat Jaume I, Department of Mathematics, Castell{\'o}n de la
  Plana, Spain, 2008, pp.~15--45.

\bibitem{Bergman:1978:APC}
{\sc D.~J. Bergman}, {\em Analytical properties of the complex effective
  dielectric constant of a composite medium with applications to the derivation
  of rigorous bounds and to percolation problems}, in Electrical Transport and
  Optical Properties of Inhomogeneous Media, J.~C. Garland and D.~B. Tanner,
  eds., vol.~40 of AIP Conference Proceedings, Woodbury, New York, 1978,
  American Institute of Physics, pp.~46--61.

\bibitem{Bernland:2010:SRC}
{\sc A.~Bernland, A.~Luger, and M.~Gustafsson}, {\em Sum rules and constraints
  on passive systems}, Technical Report LUTEDX/(TEAT-7193)/1-31/(2010), 2010.

\bibitem{Bernland:2011:SRC}
{\sc A.~Bernland, A.~Luger, and M.~Gustafsson}, {\em Sum rules and constraints
  on passive systems}, Journal of Physics A: Mathematical and Theoretical, 44
  (2011), p.~145205.

\bibitem{Bonifasi:2009:Aar}
{\sc C.~Bonifasi-Lista, E.~Cherkaev, and Y.~Yeni}, {\em Analytical approach to
  recovering bone porosity from effective complex shear modulus.}, Journal of
  biomechanical engineering, 13 (2009), p.~121003.

\bibitem{BonnetBenDhia:2012:TIP}
{\sc A.~{Bonnet-BenDhia}, L.~Chesnel, and P.~C. Jr.}, {\em T-coercivity for
  scalar interface problems between dielectrics and metamaterials},
  Mathematical Modelling and Numerical Analysis, 46 (2012), pp.~1363--1387.

\bibitem{Bouchitte:2010:CSO}
{\sc G.~Bouchitt{\'e} and B.~Schweizer}, {\em Cloaking of small objects by
  anomalous localized resonance}, Quarterly Journal of Mechanics and Applied
  Mathematics, 63 (2010), pp.~437--463.

\bibitem{Bruno:2007:SCS}
{\sc O.~P. Bruno and S.~Lintner}, {\em Superlens-cloaking of small dielectric
  bodies in the quasistatic regime}, Journal of Applied Physics, 102 (2007),
  p.~124502.

\bibitem{Cassier:2016:STM}
{\sc M.~Cassier, C.~Hazard, and P.~Joly}, {\em Spectral theory for maxwell's
  equations at the interface of a metamaterial. {Part I: Generalized Fourier
  transform.}}, available online on Arxiv at
  https://128.84.21.199/abs/1610.03021,  (2016).

\bibitem{Cessenat:1996:MME}
{\sc M.~Cessenat}, {\em Mathematical Methods in Electromagnetism: Linear Theory
  and Applications}, vol.~41 of Series on advances in mathematics for applied
  sciences, World Scientific Publishing Co., Singapore~/ Philadelphia~/ River
  Edge, New Jersey, 1996.

\bibitem{Christodoulides:2004:GVD}
{\sc Y.~T. Christodoulides and D.~B. Pearson}, {\em Generalized value
  distribution for {Herglotz} functions and spectral theory.}, Mathematical
  Physics, Analysis and Geometry, 7 (2004), pp.~309--331.

\bibitem{Dautray:1974:MAN}
{\sc R.~Dautray and J.~L. Lions}, {\em Mathematical Analysis and Numerical
  Methods for Science and Technology: {Volume} 1 {Physical} Origins and
  Classical Methods.}, Springer-Verlag, Berlin, 2000.

\bibitem{Dolin:1961:PCT}
{\sc L.~S. Dolin}, {\em To the possibility of comparison of three-dimensional
  electromagnetic systems with nonuniform anisotropic filling}, Izvestiya
  Vysshikh Uchebnykh Zavedeni{\u{\i}}. Radiofizika (see
  \url{http://www.math.utah.edu/~milton/DolinTrans2.pdf} for an english
  translation of Dolin's paper), 4 (1961), pp.~964--967.

\bibitem{Evans:2008:PDE}
{\sc L.~Evans}, {\em Partial Differential Equations}, American Mathematical
  Society, Providence, 2008.

\bibitem{Gesztesy:2000:MVH}
{\sc F.~Gesztesy and E.~Tsekanovskii}, {\em On matrix-valued {Herglotz}
  functions}, Mathematische Nachrichten, 218 (2000), pp.~61--138.

\bibitem{Golden:1983:BEP}
{\sc K.~Golden and G.~Papanicolaou}, {\em Bounds for effective parameters of
  heterogeneous media by analytic continuation}, Communications in Mathematical
  Physics, 90 (1983), pp.~473--491.

\bibitem{Greenleaf:2009:CDE}
{\sc A.~Greenleaf, Y.~Kurylev, M.~Lassas, and G.~Uhlmann}, {\em Cloaking
  devices, electromagnetic wormholes, and transformation optics}, SIAM Review,
  51 (2009), pp.~3--33.

\bibitem{Greenleaf:2003:ACC}
{\sc A.~Greenleaf, M.~Lassas, and G.~Uhlmann}, {\em Anisotropic conductivities
  that cannot be detected by {EIT}}, Physiological Measurement, 24 (2003),
  pp.~413--419.

\bibitem{Greenleaf:2003:NCI}
\leavevmode\vrule height 2pt depth -1.6pt width 23pt, {\em On non-uniqueness
  for {Calder{\'o}n}'s inverse problem}, Mathematical Research Letters, 10
  (2003), pp.~685--693.

\bibitem{Vasquez:2009:AEC}
{\sc F.~{Guevara Vasquez}, G.~W. Milton, and D.~Onofrei}, {\em Active exterior
  cloaking for the {$2$D} {Laplace} and {Helmholtz} equations}, Physical Review
  Letters, 103 (2009), p.~073901.

\bibitem{Vasquez:2009:BEC}
\leavevmode\vrule height 2pt depth -1.6pt width 23pt, {\em Broadband exterior
  cloaking}, Optics Express, 17 (2009), pp.~14800--14805.

\bibitem{Vasquez:2012:MAT}
\leavevmode\vrule height 2pt depth -1.6pt width 23pt, {\em Mathematical
  analysis of the two dimensional active exterior cloaking in the quasistatic
  regime}, Analysis and Mathematical Physics, 2 (2012), pp.~231--246.

\bibitem{Gustafsson:2010:SRP}
{\sc M.~Gustafsson and D.~Sj{\"o}berg}, {\em Sum rules and physical bounds on
  passive metamaterials}, New Journal of Physics, 12 (2010), p.~043046.

\bibitem{Gustafsson:2010:TDS}
{\sc M.~Gustafsson and D.~{Sj\"{o}berg}}, {\em Time-domain approach to the
  forward scattering sum rule}, Proceedings of the Royal Society A, 466 (2010),
  pp.~579–--3592.

\bibitem{Hashemi:2012:DBP}
{\sc H.~Hashemi, C.-W. Qiu, A.~P. McCauley, J.~D. Joannopoulos, and S.~G.
  Johnson}, {\em Diameter-bandwidth product limitation of isolated-object
  cloaking}, Physical Review A, 86 (2012), p.~013804.

\bibitem{Henrici:1993:ACC}
{\sc P.~Henrici}, {\em Applied and computational complex analysis, discrete
  {Fourier} analysis, {Cauchy} integrals, construction of conformal maps,
  univalent functions}, vol.~3, John Wiley \& Sons, 1993.

\bibitem{Jackson:1999:CE}
{\sc J.~D. Jackson}, {\em Classical Electrodynamics}, John Wiley and Sons, New
  York, NY, third~ed., 1999.

\bibitem{Kang:2008:SPS}
{\sc H.~Kang and G.~W. Milton}, {\em Solutions to the {P{\'o}lya--Szeg{\H{o}}}
  conjecture and the {Weak Eshelby Conjecture}}, Archive for Rational Mechanics
  and Analysis, 188 (2008), pp.~93--116.

\bibitem{Kato:1995:PTO}
{\sc T.~Kato}, {\em Perturbation Theory for Linear Operators}, Classics in
  Mathematics, Springer-Verlag, Berlin, Germany~/ Heidelberg, Germany~/ London,
  UK~/ etc., 1995.

\bibitem{Kerker:1975:IB}
{\sc M.~Kerker}, {\em Invisible bodies}, Journal of the Optical Society of
  America, 65 (1975), pp.~376--379.

\bibitem{Kettunen:2014:AEA}
{\sc H.~Kettunen, M.~Lassas, and P.~Ola}, {\em On absence and existence of the
  anomalous localized resonance without the quasi-static approximation},
  (2014).
\newblock Submitted. Available as arXiv:1406.6224 [math-ph].

\bibitem{Kohn:2012:VPC}
{\sc R.~V. Kohn, J.~Lu, B.~Schweizer, and M.~I. Weinstein}, {\em A variational
  perspective on cloaking by anomalous localized resonance}, Communications in
  Mathematical Physics, 328 (2014), pp.~1--27.
\newblock Available as arXiv:1210.4823 [math.AP].

\bibitem{Kohn:1984:IUC}
{\sc R.~V. Kohn and M.~S. Vogelius}, {\em Inverse problems}, in {Proceedings of
  the Symposium in Applied Mathematics of the American Mathematical Society and
  the Society for Industrial and Applied Mathematics, New York, April 12--13,
  1983}, D.~W. McLaughlin, ed., vol.~14 of SIAM AMS Proceedings, Providence,
  RI, USA, 1984, American Mathematical Society, pp.~113--123.

\bibitem{Lai:2009:CMI}
{\sc Y.~Lai, H.~Chen, Z.-Q. Zhang, and C.~T. Chan}, {\em Complementary media
  invisibility cloak that cloaks objects at a distance outside the cloaking
  shell}, Physical Review Letters, 102 (2009), p.~093901.

\bibitem{Landau:1984:ECM}
{\sc L.~D. Landau, E.~M. Lifshitz, and L.~P. Pitaevski{\u\i}}, {\em
  Electrodynamics of Continuous Media}, vol.~8 of Landau and Lifshitz Course of
  Theoretical Physics, Elsevier Butterworth-Heinemann, Oxford, UK, second~ed.,
  1984.

\bibitem{Leonhardt:2006:OCM}
{\sc U.~Leonhardt}, {\em Optical conformal mapping}, Science, 312 (2006),
  pp.~1777--1780.

\bibitem{Leonhardt:2009:BIN}
{\sc U.~Leonhardt and T.~Tyc}, {\em Broadband invisibility by non-{Euclidean}
  cloaking}, Science, 323 (2009), pp.~110--112.

\bibitem{Li:2016:QCA}
{\sc H.~Li, J.~Li, and H.~Liu}, {\em On quasi-static cloaking due to anomalous
  localized resonance in $\mathbb{R}^3$}, SIAM Journal on Applied Mathematics,
  75 (2016), pp.~1245--1260.

\bibitem{Lind-Johansen:2009:PLF}
{\sc {\O}.~Lind-Johansen, K.~Seip, and J.~Skaar}, {\em The perfect lens on a
  finite bandwidth}, Journal of Mathematical Physics, 50 (2009), p.~012908.

\bibitem{Mattner:2001:CDI}
{\sc L.~Mattner}, {\em Complex differentiation under the integral}, Nieuw
  Archief voor Wiskunde (Groningen), 5/2 (2001), pp.~32--35.

\bibitem{McPhedran:2009:CPR}
{\sc R.~C. McPhedran, N.-A.~P. Nicorovici, L.~C. Botten, and G.~W. Milton},
  {\em Cloaking by plasmonic resonance among systems of particles: cooperation
  or combat?}, Comptes Rendus Physique, 10 (2009), pp.~391--399.

\bibitem{Meklachi:2016:SAL}
{\sc T.~Meklachi, G.~W. Milton, D.~Onofrei, A.~E. Thaler, and G.~Funchess},
  {\em Sensitivity of anomalous localized resonance phenomena with respect to
  dissipation}, Quarterly of Applied Mathematics, 74 (2016), pp.~201--234.

\bibitem{Miller:2006:PC}
{\sc D.~A.~B. Miller}, {\em On perfect cloaking}, Optics Express, 14 (2006),
  pp.~12457--12466.

\bibitem{Miller:2014:FLE}
{\sc O.~D. Miller, C.~W. Hsu, M.~T.~H. Reid, W.~Qiu, B.~G. DeLacy, J.~D.
  Joannopoulos, M.~Solja{\v{c}}i{\'c}, and S.~G. Johnson}, {\em Fundamental
  limits to extinction by metallic nanoparticles}, Physical Review Letters, 112
  (2014), p.~123903.

\bibitem{Milton:1979:TST}
{\sc G.~W. Milton}, {\em Theoretical studies of the transport properties of
  inhomogeneous media}, Unpublished report TP/79/1, University of Sydney,
  Sydney, Australia, 1979.
\newblock Unpublished report. (Available on request from the author).

\bibitem{Milton:1981:BCP}
\leavevmode\vrule height 2pt depth -1.6pt width 23pt, {\em Bounds on the
  complex permittivity of a two-component composite material}, Journal of
  Applied Physics, 52 (1981), pp.~5286--5293.

\bibitem{Milton:2002:TC}
{\sc G.~W. Milton}, {\em The Theory of Composites}, vol.~6 of Cambridge
  Monographs on Applied and Computational Mathematics, Cambridge University
  Press, Cambridge, UK, 2002, pp.~295--298.
\newblock Series editors: P. G. Ciarlet, A. Iserles, Robert V. Kohn, and M. H.
  Wright.

\bibitem{Milton:1997:FFR}
{\sc G.~W. Milton, D.~J. Eyre, and J.~V. Mantese}, {\em Finite frequency range
  {Kramers-Kronig} relations: {Bounds} on the dispersion}, Physical Review
  Letters, 79 (1997), pp.~3062--3065.

\bibitem{Milton:2006:CEA}
{\sc G.~W. Milton and N.-A.~P. Nicorovici}, {\em On the cloaking effects
  associated with anomalous localized resonance}, Proceedings of the Royal
  Society A: Mathematical, Physical, \& Engineering Sciences, 462 (2006),
  pp.~3027--3059.

\bibitem{Milton:2008:SFG}
{\sc G.~W. Milton, N.-A.~P. Nicorovici, R.~C. McPhedran, K.~Cherednichenko, and
  Z.~Jacob}, {\em Solutions in folded geometries, and associated cloaking due
  to anomalous resonance}, New Journal of Physics, 10 (2008), p.~115021.

\bibitem{Mon:2014:PBE}
{\sc F.~Monticone and A.~Al\'u}, {\em Physical bounds on electromagnetic
  invisibility and the potential of superconducting cloaks}, {Photonics and
  Nanostructures - Fundamentals and Applications, Special issue for
  metamaterials}, 12 (2014), pp.~330--339.

\bibitem{Mon:2016:IEP}
\leavevmode\vrule height 2pt depth -1.6pt width 23pt, {\em Invisibility
  exposed: physical bounds on passive cloaking}, Optica, 3 (2016),
  pp.~718--724.

\bibitem{Nedelec:2001:AEE}
{\sc J.-C. Nedelec}, {\em Acoustic and Electromagnetic Equations: Integral
  Representations for Harmonic Problems}, vol.~144 of Applied Mathematical
  Sciences, Springer Science \& Business Media, New York, NY, 2001.

\bibitem{Nevanlinna:1922:AEF}
{\sc R.~Nevanlinna}, {\em Asymptotische {Entwicklungen das Stieltjessche
  Momentenproblem}}, Annales Academiae Scientiarum Fennicae, Series A, 18
  (1922).

\bibitem{Nguyen:2015:CAL}
{\sc H.-M. Nguy{\^e}n}, {\em Cloaking via anomalous localized resonance for
  doubly complementary media in the quasistatic regime}, Journal Of The
  European Mathematical Society, 17 (2015), pp.~1327--1365.

\bibitem{Nguyen:2016:EWP}
\leavevmode\vrule height 2pt depth -1.6pt width 23pt, {\em Cloaking an
  arbitrary object via anomalous localized resonance: the cloak is independent
  of the object.},  (2016).
\newblock Available as arXiv:1607.06492.

\bibitem{Nguyen:2016:CVA}
\leavevmode\vrule height 2pt depth -1.6pt width 23pt, {\em Cloaking via
  anomalous localized resonance for doubly complementary media in the finite
  frequency regime},  (2016).
\newblock Available as arXiv:1511.08053 [math.AP].

\bibitem{Nguyen:2016:LAP}
{\sc H.-M. Nguyen}, {\em Limiting absorption principle and well-posedness for
  the {Helmholtz} equation with sign changing coefficients}, Journal de
  Math{\'e}matiques Pures et Appliqu{\'e}es, 106 (2016), pp.~342--374.

\bibitem{Nguyen:2015:CAL1}
{\sc H.-M. Nguy{\^e}n and L.~H. Nguy{\^e}n}, {\em Cloaking using complementary
  media for the {Helmholtz} equation and a three spheres inequality for second
  order elliptic equations}, Transactions of The American Mathematical Society,
  Series B, 2 (2015), pp.~93--112.

\bibitem{Nguyen:2016:CCM}
{\sc L.~H. Nguy{\^e}n}, {\em Cloaking using complementary media in the
  quasistatic regime}, Annales de l'Institut Henri Poincar{\'e}. Analyse non
  lin{\'e}aire,  (2016).
\newblock In press. Available online.

\bibitem{Nicorovici:1994:ODP}
{\sc N.~A. Nicorovici, R.~C. McPhedran, and G.~W. Milton}, {\em Optical and
  dielectric properties of partially resonant composites}, Physical Review B
  (Solid State), 49 (1994), pp.~8479--8482.

\bibitem{Nicorovici:2011:RLD}
{\sc N.-A.~P. Nicorovici, R.~C. McPhedran, and L.~C. Botten}, {\em Relative
  local density of states and cloaking in finite clusters of coated cylinders},
  Waves in Random and Complex Media. Propagation, Scattering and Imaging, 21
  (2011), pp.~248--277.

\bibitem{Nicorovici:2008:FWC}
{\sc N.-A.~P. Nicorovici, R.~C. McPhedran, S.~Enoch, and G.~Tayeb}, {\em Finite
  wavelength cloaking by plasmonic resonance}, New Journal of Physics, 10
  (2008), p.~115020.

\bibitem{Nicorovici:2007:QCT}
{\sc N.-A.~P. Nicorovici, G.~W. Milton, R.~C. McPhedran, and L.~C. Botten},
  {\em Quasistatic cloaking of two-dimensional polarizable discrete systems by
  anomalous resonance}, Optics Express, 15 (2007), pp.~6314--6323.

\bibitem{Norris:2015:AIE}
{\sc A.~N. Norris}, {\em Acoustic integrated extinction}, Proceedings of the
  Royal Society of London. Series {A}, 471 (2015), p.~20150008.

\bibitem{Norris:2012:SAA}
{\sc A.~N. Norris, F.~A. Amirkulova, and W.~J. Parnel}, {\em Source amplitudes
  for active exterior cloaking}, Inverse Problems, 28 (2012), p.~105002.

\bibitem{Norris:2014:AEC}
{\sc A.~N. Norris, F.~A. Amirkulova, and W.~J. Parnell}, {\em Active
  elastodynamic cloaking}, Mathematics and Mechanics of Solids : MMS, 19
  (2014), pp.~603--625.

\bibitem{Nussenzveig:1972:CDR}
{\sc H.~M. Nussenzveig}, {\em Causality and dispersion relations}, Academic
  Press, New York, 1972.

\bibitem{ONeill:2015:ACI}
{\sc J.~{O'Neill}, {\"O}.~Selsil, R.~C. McPhedran, A.~B. Movchan, and N.~V.
  Movchan}, {\em Active cloaking of inclusions for flexural waves in thin
  elastic plates}, Quarterly Journal of Mechanics and Applied Mathematics, 68
  (2015), pp.~263--288.

\bibitem{ONeill:2016:ACR}
{\sc J.~{O'Neill}, {\"O}.~Selsil, R.~C. McPhedran, A.~B. Movchan, N.~V.
  Movchan, and C.~H. Moggach}, {\em Active cloaking of resonant coated
  inclusions for waves in membranes and kirchhoff plates}, Quarterly Journal of
  Mechanics and Applied Mathematics, 69 (2016), pp.~115--159.

\bibitem{Onofrei:2012:AMF}
{\sc D.~Onofrei}, {\em On the active manipulation of fields and applications:
  {I}. {The} quasistatic case}, Inverse Problems, 28 (2012), p.~105009.

\bibitem{Onofrei:2016:ALR}
{\sc D.~Onofrei and A.~E. Thaler}, {\em Anomalous localized resonance phenomena
  in the nonmagnetic, finite-frequency regime},  (2016).
\newblock Submitted. Available as arXiv:1605.08954 [math-ph].

\bibitem{Osborn:1945:DFG}
{\sc J.~A. Osborn}, {\em Demagnetizing factors of the general ellipsoid},
  Physical Review, 67 (1945), pp.~351--357.

\bibitem{Pecseli:2000:FPS}
{\sc H.~L. P\'{e}cseli}, {\em Fluctuations in physical systems.}, Cambridge
  University Press, 2000.

\bibitem{Pendry:2006:CEM}
{\sc J.~B. Pendry, D.~Schurig, and D.~R. Smith}, {\em Controlling
  electromagnetic fields}, Science, 312 (2006), pp.~1780--1782.

\bibitem{Purcell:1969:AEL}
{\sc E.~M. Purcell}, {\em On the absorption and emission of light by
  interstellar grains.}, The Astrophysical Journal, 158 (1969), pp.~433--440.

\bibitem{Slvanayagam:2012:AEC}
{\sc M.~Selvanayagam and G.~V. Eleftheriades}, {\em An active electromagnetic
  cloak using the equivalence principle}, IEEE Antennas and Wireless
  Propagation Letters, 11 (2012), pp.~1226--1229.

\bibitem{Sohl:2008:DRS}
{\sc C.~Sohl}, {\em Dispersion Relations in Scattering and Antenna Problems},
  {Ph.D.} thesis, available online at
  http://lup.lub.lu.se/search/record/1221227, Lund University, 2008.

\bibitem{Stoner:1945:DFE}
{\sc E.~C. Stoner}, {\em The demagnetizing factors for ellipsoids},
  Philosophical Magazine, 36 (1945), pp.~803--820.

\bibitem{Tip:2004:LDD}
{\sc A.~Tip}, {\em Linear dispersive dielectrics as limits of {Drude-Lorentz}
  systems}, Physical Review E (Statistical physics, plasmas, fluids, and
  related interdisciplinary topics), 69 (2004), p.~016610.

\bibitem{Veselago:1967:ESS}
{\sc V.~G. Veselago}, {\em The electrodynamics of substances with
  simultaneously negative values of $ \epsilon $ and $ \mu $}, Uspekhi
  Fizicheskikh Nauk, 92 (1967), pp.~517--526.
\newblock English translation in {\it{Soviet Physics Uspekhi}}
  \bblvolume{10}(\bblnumber{4}):509--514 (1968).

\bibitem{Welters:2014:SLL}
{\sc A.~T. Welters, Y.~Avniel, and S.~G. Johnson}, {\em Speed-of-light
  limitations in passive linear media}, Physical Review A (Atomic, Molecular,
  and Optical Physics), 90 (2014), p.~023847.

\bibitem{Yaghjian:2006:PWS}
{\sc A.~D. Yaghjian and T.~B. Hansen}, {\em Plane-wave solutions to
  frequency-domain and time-domain scattering from magnetodielectric slabs},
  Physical Review E (Statistical physics, plasmas, fluids, and related
  interdisciplinary topics), 73 (2006), p.~046608.

\bibitem{Zemanian:2005:RTC}
{\sc A.~H. Zemanian}, {\em Realizability theory for continuous linear
  systems.}, Courier Corporation, 1972.

\end{thebibliography}

\end{document}